\documentclass[epsfig,11pt]{llncs}
\usepackage{amssymb,amsmath,graphicx,color}
\usepackage{amsfonts}

\newcommand{\remove}[1]{}
\usepackage{epsfig,latexsym,graphicx}
\begin{document}

\title{Spline Smoothing for Estimation of Circular Probability Distributions via Spectral Isomorphism and its Spatial Adaptation}
\author{Kinjal Basu$^1$ \and Debapriya Sengupta$^2$}
\institute{Indian Statistical Institute\\
203, B.T. Road, Kolkata - 700 108, India\\
Email : $^1$bst0707@isical.ac.in, $^2$dps@isical.ac.in}
\maketitle

\begin{abstract}
Consider the problem when $X_1, X_2, \ldots, X_n$ are distributed on a circle following an unknown distribution $F$ on $S^1$. In this article we have consider the absolute general set-up where the density can have local features such as discontinuities and edges. Furthermore, there can be outlying data which can follow some discrete distributions. The traditional Kernel Density Estimation methods fail to identify such local features in the data. Here we device a non-parametric density estimate on $S^1$, by the use of a novel technique which we term as Fourier Spline. We have also tried to identify and incorporate local features such as support, discontinuity or edges in the final density estimate. Several new results are proved in this regard. Simulation studies have also been performed to see how our methodology works. Finally a real life example is also shown.
\end{abstract}

\section{Introduction}
Let $(\mathcal{X},\mathcal{A})$ be a metric space with a Borel $\sigma$-field and $P$ be a probability distribution on $(\mathcal{X},\mathcal{A})$. Further we assume that this space has an inherent topological structure. Let us say we have samples $X_1, X_2, \ldots, X_n$ i.i.d from $P$. Let the us denote the empirical distribution as $P_n$, where,
\begin{equation}
P_n = \frac{1}{n}\sum_{i=1}^{n}\delta_{X_i}.
\end{equation}
Note that the empiricals may be given in terms of histogram data as well, i.e. we can have data from disjoint partitions $A_1, A_2, \ldots A_k$ on the sample space, with uniform Haar measure $\lambda$. In such a situation the empirical distribution is given by
\begin{equation}
P_n = \sum_{i=1}^{k}\frac{n_i}{n}\frac{I_{A_i}(\cdot)}{\lambda(A_i)}
\end{equation}
where $n_i$ denotes the number of points in $A_i$ for $i = 1 \ldots k$.

Assuming that the data is from some unknown density $f(x)$ a natural component of exploratory data analysis is to estimate the function $f(\cdot)$. Any density estimator is a descriptor of the population distribution, and when it is known that the population distribution is absolutely continuous with respect to some standard invariant measure, the Radon-Nykodym derivative is the density. Now, the major question is what is the central problem in using empirical distribution for density estimation? Note that $\widehat{P(f)} = \int fdP_n$ is unbiased and $Var(\widehat{P(f)})$ goes to 0 as $1/n$. However, in case we assume that $P$ has a derivative $p$ with respect to the Haar measure $\lambda$ on $(\mathcal{X},\mathcal{A})$ (in the sense of the Radon-Nykodym derivative) the empirical estimate does not work. For example, it is a known that in $\mathbb{R}$, $\frac{count(a,b)}{b-a}$ does not work.

There are few different approaches in practice which are majorly used, viz., Kernel density estimation, Spline smoothing and Orthogonal Series.

\subsection{Kernel Density Estimation}
Kernel density estimation (KDE) is a non-parametric way to estimate the probability density function of a random variable. Kernel density estimation is a fundamental data smoothing problem where inferences about the population are made, based on a finite data sample. The solution to the smoothing problem lies in the space of all continuously twice differentiable functions, with square integrable second derivative.
For KDE the Reisz representation is essentially used. $P = Q$ if $\int\phi dP = \int\phi dQ$ for all bounded continuous functions $\phi$. Let $K_h(x - \cdot), x \in \mathcal{X}$ denote the family of bounded continuous functions for which the estimate is unbiased, $h$ being a smoothing parameter, i.e.
\begin{equation}
\widehat{P(K_h(x - \cdot))} = \int K_h(x - y)dP_n = P_n(K_h(x - \cdot)) = \frac{1}{n}\sum_{i=1}^{n}K_h(x - X_i)
\end{equation}
$K_h$ satisfies for bounded continuous functions,
\begin{equation}
f(x) = \lim_{h \to 0}\int K_h(x - y) f(y) d\lambda(y) = \delta_x(f)
\end{equation}
Now $d\nu = K_h(x - y)d\lambda(y)$ can be thought of as an approximation to the degenerate measure $\delta_x$ when $h$ is small. Note the $\nu << \lambda$ with $K_h(x - \cdot)$ as its Radon-Nykodym derivative. Thus (3) defines a density estimate. 
\begin{equation}
\hat{f}(x;h) = \frac{1}{n}\sum_{i=1}^{n}K_h(x - X_i)
\end{equation}
When $\mathcal{X} = \mathbb{R}$, we have the following results from Silverman (1986) \cite{Silverman} ,
\begin{eqnarray}
\textnormal{Bias}(\hat{f}(x)) &\approx & \frac{h^2}{2}f{''}(x)\int_{-\infty}^{\infty}z^2K(z)\;dz\\
\textnormal{Var}(\hat{f}(x)) &\approx & \frac{1}{nh}f(x)\int_{-\infty}^{\infty}K^2(z)dz\\
\textnormal{MSE}(\hat{f}(x)) &\approx & \frac{h^4}{4}f{''}(x)^2\left(\int_{-\infty}^{\infty}z^2K(z)\;dz\right)^2 + \frac{1}{nh}f(x)\int_{-\infty}^{\infty}K^2(z)dz
\end{eqnarray}
Integrating with respect to $x$, we get
\begin{equation}
\textnormal{MISE}(\hat{f}) \approx \frac{1}{4}h^4k_2^2\beta(f) + \frac{1}{nh}j_2
\end{equation}
where $K_h(x) = K(x/h)/h,\;\; k_2 = \int z^2K(z)dz,\;\; \beta(f) = \int f''(x)^2dx$ and $j_2 = \int K^2(z)dz$. Based on this the optimal value of the bandwidth is derived as 
\begin{equation}
h_{opt} = \left(\frac{1}{n} \frac{j_2k_2^{-2}}{\beta(f)}\right)^{1/5}. 
\end{equation}
Thus, the minimal MISE can be shown to be
\begin{equation}
MISE_{opt}(\hat{f}) = \frac{5}{4}\left(\frac{\beta(f)j_2^4k_2^2}{n^4}\right)^{1/5}
\end{equation}

\subsection{Spline Smoothing}
\subsubsection{1.2.1 Classical Smoothing Spline Background : \\}
The usual smoothing spline problem \cite{Grace} can be considered as a minimization problem. Let $(x_i,Y_i);\;\; x_1<x_2<\ldots<x_n, i \in \mathbb{Z}$ be a sequence of observations, modelled by the relation $Y_i = \mu(x_i)$. The smoothing spline estimate $\hat{\mu}$ of the function $\mu$ is defined to be the minimizer \cite{hastie} (over the class of twice differentiable functions) of 

\begin{equation}
\sum_{i=1}^n (Y_i - \hat\mu(x_i))^2 + \lambda \int_{x_1}^{x_n} \hat\mu''(x)^2 \,dx. 
\end{equation}
Note that
\begin{itemize}
\item $\lambda \ge 0$ is a smoothing parameter, controlling the trade-off between fidelity to the data and roughness of the function estimate.
\item  The integral is evaluated over the range of the $x_i$.
\item As $\lambda\to 0$ (no smoothing), the smoothing spline converges to the interpolating spline.
\item As $\lambda\to\infty$ (infinite smoothing), the roughness penalty becomes paramount and the estimate converges to a linear least squares estimate.
\end{itemize}

\subsubsection {1.2.2 Typical Solution to the Smoothing Problem\\}
It is useful to think of fitting a smoothing spline in two steps:
\begin{itemize}
\item First, derive the values $\hat\mu(x_i);i=1,\ldots,n.$
\item From these values, derive $\hat\mu(x)$ for all $x$.
\end{itemize}
   
Now, consider the second step first.

Given the vector $\hat{m} = (\hat\mu(x_1),\ldots,\hat\mu(x_n))^T$ of fitted values, the sum-of-squares part of the spline criterion is fixed. It remains only to minimize $\int \hat\mu''(x)^2 \, dx$, and the minimizer is a natural cubic spline that interpolates the points $(x_i,\hat\mu(x_i))$. This interpolating spline is a linear operator, and can be written in the form
\begin{equation}
    \hat\mu(x) = \sum_{i=1}^n \hat\mu(x_i) f_i(x) 
\end{equation}
where $f_i(x)$ are a set of spline basis functions. As a result, the roughness penalty has the form
\begin{equation}
\int \hat\mu''(x)^2 dx = \hat{m}^T A \hat{m}. 
\end{equation}
where the elements of A are $\int f_i''(x) f_j''(x)dx$. The basis functions, and hence the matrix A, depend on the configuration of the predictor variables $x_i$, but not on the responses $Y_i$ or $\hat m$.

Now going back to the first step, the penalized sum-of-squares can be written as
\begin{equation}
    \|Y - \hat m\|^2 + \lambda \hat{m}^T A \hat m, 
\end{equation}
where $Y=(Y_1,\ldots,Y_n)^T$. Minimizing over $\hat m$ gives
\begin{equation}
    \hat m = (I + \lambda A)^{-1} Y. 
\end{equation}

\subsection{Orthogonal Series}
It helps in approximations in a suitable Hilbert space $\mathcal{H}$ of functions containing smooth functions. However, since the original probability distribution is unknown, the choice of $\mathcal{H}$ becomes arbitrary. Different techniques are needed for different spaces, for eg. in Sobolev space [GIVE REF], we take an orthonormal basis $\lbrace u_m \rbrace_{m \in \mathbb{Z}}$ of $\mathcal{H}$.
\begin{equation}
\widehat{P(u_m)} = P_n(u_m) = \frac{1}{n}\sum_{i=1}^{n}u_m(X_i)
\end{equation}
Note that here using feature transformation we can transform the probability into the space of sequences, i.e. $\mathbb{P}_n \longrightarrow \left(\mathbb{P}_n(u_m) : m \in \mathbb{Z}\right)$. Although it is true that $\sum P(u_m) u_m$ is convergent in $\mathcal{H}$, the formal Fourier transformation, $\sum \mathbb{P}_n(u_m) u_m$ does not converge in $\mathcal{H}$. Mimicking the KDE approach, the smoothing here requires the multiplication of the empirical coefficients by a smoothing sequence $\lbrace c_m \rbrace$ such that $\sum c_m^2 |\mathbb{P}_m(u_m)|^2 < \infty$.

Now $P(\phi)$ with $\phi = \sum \langle u_m, \phi \rangle u_m$ can be written as 
\begin{equation}
P(\phi) = \int \phi dP = \int \phi \frac{dP}{d\lambda}d\lambda = \sum \langle u_m, \phi \rangle \int u_m \frac{dP}{d\lambda}d\lambda
\end{equation}
Thus a natural estimate of $P(\phi)$ is then
\begin{equation}
\widehat{P(\phi)} = \sum \langle u_m, \phi \rangle \mathbb{P}_n(u_m)c_m
\end{equation}
Now by the Cauchy-Schwarz inequality the right hand side is convergent in $\mathcal{H}$. Hence the orthogonal series estimator of $dP/d\lambda$ is extracted by applying uniqueness of inverse Fourier transformation [GIVE REFS] to the coefficients $\lbrace \mathbb{P}_n(u_m)c_m \rbrace$ i.e.
\begin{equation}
\widehat{\frac{dP}{d\lambda}} = \sum \mathbb{P}_n(u_m)c_m u_m
\end{equation}

\subsubsection{Remarks : }
The choice of the Hilbert space is unknown. There are several different Hilbert spaces and basis which can be used, such as the Fourier basis, and wavelets based on unknown smoothness parameters and the choice of the thresholding sequence $\lbrace c_m \rbrace$ \cite{Donoho}. Due to the arbitrary-ness in the choice of Hilbert space, the choice of $\lbrace c_m \rbrace$, its efficiency and error estimate is user driven.

The orthogonal series estimate are also convolution estimates $\frac{1}{n}\sum \tilde{K}(x - X_i)$ where $\tilde{K}(x) = \sum c_mu_m(x)$. This holds by the convolution property of the Fourier transformation.

There is also an issue regarding the sample space itself. If $(\mathcal{X},\mathcal{A})$ does not have a nice topological group structure, (for eg. manifolds in $\mathbb{R}^n$, etc), then the question of estimation of density is not well posed because we do not have a base measure with respect to which we can define the density. In such cases we look for a transformation $T$ which is a 1-1 measurable open mapping that maps $(\mathcal{X},\mathcal{A})$ into a locally compact topological group. Since there exists a natural Haar measure $\lambda$ in the topological group, it makes sense to talk about density in this situation. Note that $T$ must satisfy $\lambda(T(\mathcal{X})^c) = 0$. If we denote the density in this space by $g_T = \frac{dPoT^{-1}}{d\lambda}$, we can estimate all integrals of the form $\int h(T)dP$ by 
\begin{equation}
\widehat{\int h(T)dP} = \int h(y)\hat{g}_T(y) d\lambda(y)
\end{equation}
Such functions generate the $\sigma$-field $\sigma(T)$. If $\sigma(T) = \mathcal{B}$ (Borel $\sigma$-field), i.e., if $T$ is a 1-1 isomorphic mapping, we can find all such integrals in this indirect fashion.

\subsection{Focus on estimation on $S^1$ with respect to the Haar measure}
The major reasons for concentrating on $S^1$ are the as follows
\begin{itemize}
\item Trigonometric Fouries basis are computationally very efficient, with tools such as the Fast Fourier Transformations, etc
\item The Bochner's Theorem which states that,
\begin{theorem}
Every positive definite function $Q$ is the Fourier transform of a positive finite Borel measure.
\end{theorem}
\begin{proof}
Let $F_0(\mathbb{R})$ be the family of complex valued functions on $\mathbb{R}$ with finite support, i.e. $f(x) = 0$ for all but finitely many $x$. The positive definite kernel $K(x,y)$ induces a sesquilinear form on $F_0(\mathbb{R})$. This in turn results in a Hilbert space $\left( \mathcal{H}, \langle \;,\; \rangle \right)$
whose typical element is an equivalence class $[g]$. For a fixed $t$ in $\mathbb{R}$, the ``shift operator" $U_t$ defined by $(U_tg)(x) = g(x - t)$, for a representative of $[g]$ is unitary. In fact the map
$t \; \stackrel{\Phi}{\mapsto} \; U_t$ is a strongly continuous representation of the additive group $\mathbb{R}$. By Stone's theorem, there exists a (possibly unbounded) self-adjoint operator $A$ such that $U_{-t} = e^{-iAt}.\;$ This implies there exists a finite positive Borel measure $\mu$ on $\mathbb{R}$ where
$\langle U_{-t} [e_0], [e_0] \rangle = \int e^{-iAt} d \mu(x)$,
where $e_0$ is the element in $F_0(\mathbb{R})$ defined by $e_0(m) = 1$ if $m = 0$ and $0$ otherwise. Because
$\langle U_{-t} [e_0], [e_0] \rangle = K(-t,0) = Q(t),$
the theorem holds.
\begin{flushright}$\blacksquare$\end{flushright}
\end{proof}
The Bochner's Theorem gives a 1-1 correspondence between non-negative definite probability measure and non-negative definite sequences.
\item The technique of Fourier spline. Choosing the orthogonal sequence as $\lbrace e^{inx} \rbrace$ satisfies other nice properties in terms of homomorphisms, error in approximations, etc. The choice of the penalty can also be easily be determined. 
\end{itemize}

The main criticism of using the Fourier basis is the fact it fails to capture local features such as discontinuities and edges. 

A survey of the literature shows several works on the adaptation to unknown smoothness by methods such as wavelet shrinkage (Donoho and Johnstone (1995)) \cite{Donoho2}, aggregation of thresholded wavelet estimators (Chesneau and Lecue (2009)) \cite{Chesneau}. Several other methods include, the cross-validation methods (Nason (1995) \cite{Nason} and Jansen (2001) \cite{Jansen}), the methods based on hypothesis tests (Abramovich, Benjamini, Donoho and Johnstone (2006) \cite{Abra}), the Lepski methods (Juditsky (1997) \cite{Jud}) and the Bayesian methods (Abramovich, Sapatinas and Silverman (1998) \cite{Abra2}).

    Edge preserving function estimation has also been studied in literature especially by MAP estimators,(see Bouman and Sauer (1992) \cite{Bouman}).

\paragraph{}
In this paper we develop a technique for estimation of densities on $S^1$ using Fourier Spline and address the issue of adapting non-smooth local features using a hybrid of frequency and spatial domain. Few results on the kernel density estimation can be found in Pelletier (2005) \cite{Pelletier} and Taylor (2008) \cite{Taylor}.

The classical kernel density estimator was first proposed by Fisher \cite{Fisher} for data lying on the circle, in which he adapted linear data methods of Silverman \cite{Silverman} and used a quartic kernel function $K(\theta) = 0.9375(1-\theta^2)^2$. However, when using data on the circle, we cannot use distance in Euclidean space, so all differences $\theta - \theta_i$ should be replaced by considering the angle between two vectors:
\begin{equation}
d_i(\theta) = ||\theta - \theta_i|| = \min(|\theta - \theta_i|, 2\pi - |\theta - \theta_i|)
\end{equation}

This may also be written as $d_i = \cos^{-1}(\boldsymbol{x'x_i})$ where $\boldsymbol{x'} =
(\cos \theta, \sin \theta)$ is a unit vector. A more natural choice for the kernel function is therefore one of the commonly used circular probability densities, such as the wrapped normal distribution, or the von Mises distribution. This leads to an alternative representation for the kernel density estimate \cite{Sengupta}
\begin{equation}
\hat{f}(\boldsymbol{\theta};h) = \frac{1}{n}\sum_{i=1}^{n}K_h(1-\boldsymbol{x'x_i})
\end{equation}

In studying properties of kernel density estimates in Euclidean space, it is common to take Taylor series approximations to give an asymptotic form for the bias and variance. These can then be combined to give an asymptotically optimal choice for the smoothing parameter; see, for example, \cite{Silverman}. For data lying on the $q$ - dimensional sphere $(q \geq 2)$ \cite{Hall} described the asymptotic bias and variance of two classes of kernel estimators. This was done by the use of directional derivatives, thus making the results a close analogue of the Taylor series methods used for data in Euclidean space.

However, the classical kernel density estimators does not give a good estimate in a general set-up, when data might follow mixture distribution of a continuous and discrete distribution with non intersecting support. In order to illustrate, let us consider the following. Suppose $X_1, X_2, \ldots X_n \sim N(0,1)$ truncated on $[-\frac{\pi}{2},\frac{\pi}{2})$ with probability $1 - \epsilon$ and a discrete distribution with probability $\epsilon$, where the discrete distribution takes the values $\frac{3\pi}{4}$ and $-\frac{3\pi}{4}$ with equal probability. Note that this discrete distribution can be viewed as outlying data. The mean squared error using the usual kernel density is estimated via simulation study. The simulations have been performed using different sample sizes and by taking $\epsilon = 0.01, 0.05$ and $0.1$. The percentage increase in the MISE when compared to the case taking $\epsilon = 0$ is shown is Table 1. The whole process is repeated $N = 10^5$ times.

\begin{table}[h]
\centering
\caption{Mean Integrated Square Error of Classical Kernel Estimate}
\begin{tabular}{|@{\hspace{0.25cm}}c@{\hspace{0.25cm}}|@{\hspace{0.25cm}}c@{\hspace{0.25cm}}|@{\hspace{0.25cm}}c@{\hspace{0.25cm}}|}
\hline
Sample Size & Value & Percentage \\
$n$ &  of $\epsilon$ & increase in MISE \\
\hline
&0.01&14.86\%\\
100&0.05&47.68\%\\
&0.10&65.68\%\\
\hline
&0.01&14.90\%\\
500&0.05&57.95\%\\
&0.10&79.13\%\\
\hline
&0.01&14.95\%\\
1000&0.05&60.26\%\\
&0.10&80.77\%\\
\hline
\end{tabular}
\end{table}

Note that the mean square error increases almost $15\%$ when only $ 1\%$ outlying data is present. For $10\%$ outlying data, the mean square error for the classical kernel estimate increases about $81\%$. The Figure 1 shows  the estimated density for the different $\epsilon$ values. Note that for $\epsilon = 0.10$, the high discrete nature of the data is being captured by a large smooth region instead of peaks at the two specified points. This raises the MSE of the Kernel density estimate as seen in Table 1. This motivates us to work on such kind of problems when he density is actually a mixture of a smooth and discrete family. We have adapted a method based on Splines using Fourier techniques, to work on such problems. Several results have been proved in this regard. A comparative study has also been performed with other known techniques. Optimality conditions and  optimal choice of smoothing parameters are also calculated theoretically. The detailed methodology is explained in Section 2. 

\begin{figure}[h]
\centering
\includegraphics[scale=0.55]{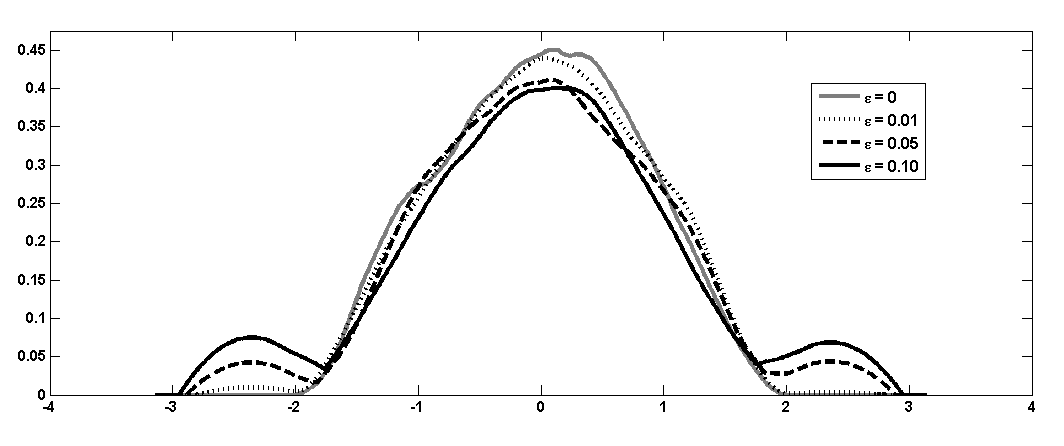}
\caption{Kernel Density Estimates for Different $\epsilon$ Values}
\end{figure}

It is also important to note that the discrete distribution with a disjoint support could actually be interpreted as some arbitrary outliers in the situation when we are dealing with circular data. Thus when dealing with such outliers two major aspects come to the limelight, viz. detection and accommodation. In this article we have developed a novel method of dealing with such outliers and accommodating them inherently. Details on the procedure are explained later in Section 3.2.

It is also particularly difficult to detect local features while using such kernel density estimators. Local features such as boundary points of support, points of discontinuity or the presence of sharp edges in the density or both, are not at all easy to identify, since during smoothing, such local features are lost. We have developed a technique to identify such local features and adapt them accordingly in the final estimate. Details are provided later in Section 3.3. The figure below shows a density drawn from a mixture of a uniform density and a triangular density, showing presence of both discontinuity and edges.

\begin{figure}[h]
\centering
\includegraphics[scale=0.55]{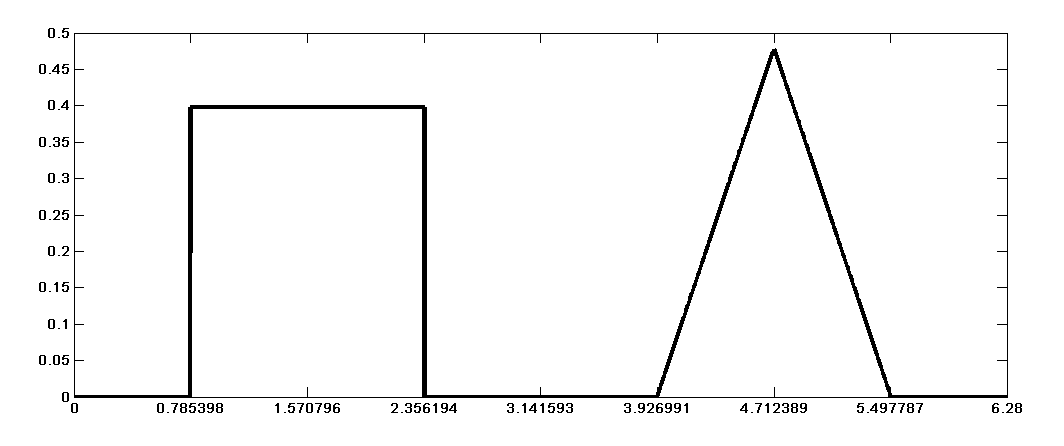}
\caption{Example of a Density with local features such as discontinuity at $\frac{\pi}{4}$ and $\frac{3\pi}{4}$ and edges at $\frac{5\pi}{4}$ and $\frac{7\pi}{4}$}
\end{figure}

With the knowledge of such points of local feature, we shall be able to break the circle into small compact disjoint sets which over which the density is smooth. So we adapt our methodology of Fourier Splines to estimate the density in these compact sets. Finally at the points of detected localizations we estimate the local features using a two parameter exponential model. Finally, we join these two estimates to get the final estimate, based on a method in general topology called the `Partition of Unity' \cite{Munkres}

\subsection{Using Partitions of Unity}

A partition of unity of a topological space $\mathcal{X}$ is a set of continuous functions, $\{\rho_i\}_{i\in I}$, from $\mathcal{X}$ to the unit interval [0,1] such that for every point, $x\in X$,
\begin{itemize}
\item there is a neighbourhood of $x$ where all but a finite number of the functions are 0, and
\item the sum of all the function values at $x$ is 1, i.e., $\sum_{i\in I} \rho_i(x) = 1$.
\end{itemize}

Sometimes, the requirement is not as strict: the sum of all the function values at a particular point is only required to be positive rather than a fixed number for all points in the space.
Partitions of unity are useful because they often allow one to extend local constructions to the whole space. 

The existence of partitions of unity assumes two distinct forms:

\begin{itemize}
\item Given any open cover $\{U_i\}_{i \in I}$ of a space, there exists a partition $\{\rho_i\}_{i \in I}$ indexed over the same set $I$ such that supp $\rho_i \subseteq U_i$. Such a partition is said to be subordinate to the open cover $\{U_i\}_i$.
\item Given any open cover $\{U_i\}_{i \in I}$ of a space, there exists a partition $\{\rho_j\}_{j \in J}$  indexed over a possibly distinct index set $J$ such that each $\rho_j$ has compact support and for each $j \in J$, supp $\rho_j \subseteq U_i$ for some $i \in I$.
\end{itemize}
Thus one chooses either to have the supports indexed by the open cover, or the supports compact. If the space is compact, then there exist partitions satisfying both requirements.

The construction uses mollifiers (bump functions), which exist in the continuous and smooth manifold categories, but not the analytic category. Thus analytic partitions of unity do not exist. 

\begin{figure}
\centering
\includegraphics[scale=0.15]{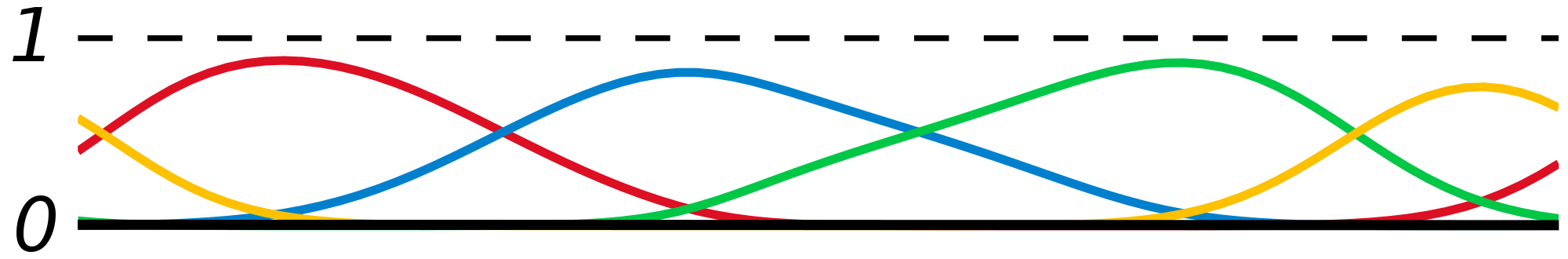}
\caption{A partition of unity of a circle with four functions. The circle is unrolled to a line segment (the bottom solid line) for graphing purposes. The dashed line on top is the sum of the functions in the partition.}
\end{figure}

For our use, partition of unity helps us in localizing the problem. Based on this, any smooth function $f$ can be decomposed as
\[\begin{split}
f &=  \left\lbrace\sum_i\rho_i + \left( 1 - \sum_i\rho_i\right)\right\rbrace f \\
  &= \sum_i\rho_if + \left(1-\sum_i\rho_i\right) f
\end{split}
\]
where $\left(1-\sum_i\rho_i\right) f$ denotes the smooth part and $\sum_i\rho_if$ covers the localization. We treat these two separately and finally add to get the final estimate such that the final estimate reflects the localization of the problem adequately. 
The novelty of our approach lies on the fact that along with the smooth estimate of the density, various local and other features are also shown. Instead of just a vector based output, a complete structure is obtained, which enables the data to be stored in a structural form. Thus being of utmost importance in situations pertaining to database handling.

\subsubsection{Remarks : }
Estimation of $ \rho_1 f, \rho_2 f, \ldots $ are usually based on wavelets. We have used a local adaptation of the wavelet procedure using an appropriate basis. Several methods of such adaptation are available in literature, especially the work by Donoho and Johnstone \cite{Donoho}. We have used local exponential model to estimate these functions, while the Fourier Spline technique is used to estimate the smooth function. The window for $\rho_1$ and $\rho_2$, ... (more if we detect more points) are determined experimentally. Note that these are anyway centred at the detected locations of discontinuities. We apply this technique on several simulated data, and the results are shown in Section 5.

The rest of the article is organized as follows; In Section 2, we focus on the methodology for density estimation using the Fourier spline approach, giving details on MSE and penalty selection. In Section 3 we adapt our method to local features (explaining the various local features) and the strategy to handle (i) support and outlier issue (ii) local exponential modelling of jump discontinuities and edges. Section 4 deals with the concept of unifying the local estimates and the smooth function estimates through an idea of partition of unity. Simulated study, real life data and discussions follow in Sections 5, 6 and 7. Finally we  conclude in Section 8.

\section{Smooth Density Estimation Using Fourier Spline}
	In our situation we are working with observations on a circle which are assumed to follow a probability distribution which is absolutely continuous with respect to the Lebesgue measure or the Haar measure. And our aim is to find its density estimate. There are several methods discussed in literature \cite{Sengupta} regarding the estimation of density on the circle such as the Kernel Density Estimates with detailed study on their Bandwidth selection \cite{Taylor}.
		
	In this article, we propose a method to use the Fourier Splines technique to get an estimate of the density at a point $x$. Comparative study has been performed between our procedure and the usual procedures based on Kernel Density Estimation. We have used Cross-validation \cite{Chiu} and Plug-in \cite{Loader} methods in our comparative study. We have not only used the usual Epanechnikov kernel but also a custom kernel as $K(\theta) \propto \cos^2(\theta)$ for $\theta \in [-\frac{\pi}{2},\frac{\pi}{2})$. The Mean Square Error studies are shown in Table 2. Finally we have incorporated a penalty selection method in the procedure and found its optimal estimate using the bias-variance trade-off calculations. Repeated simulation study has been performed to computationally prove the accuracy of our selection estimates.

\subsection{Fourier Series on Circles}
Before we begin the detailed theory, let us give a brief introduction to the General Fourier Transformation theory.

\subsubsection{2.1.1. Classical Fourier Transformation \cite{Feller} : }
There are several common conventions for defining the Fourier transform $\hat{f}$ of an integrable function $f : \mathbb{R} \to \mathbb{C}$ .In this article will use the definition:

\begin{equation}
\hat{f}(\xi) = \int_{-\infty}^{\infty} f(x)\ e^{2\pi i x \xi}\,dx,   \;\; \textnormal{ for every real number } \xi.
\end{equation}    

When the independent variable x represents time, the transform variable $\xi$  represents frequency. Under suitable conditions, $
f$ can be reconstructed from $\hat{f}$ by the inverse transform:

\begin{equation}
f(x) = \int_{-\infty}^{\infty} \hat{f}(\xi)\ e^{- 2 \pi i \xi x}\,d\xi,\;\; \textnormal{ for every real number } x.
\end{equation}

The statement that $f$ can be reconstructed from $\hat{f}$ is known as the Fourier integral theorem, and was first introduced in Fourier's Analytical Theory of Heat (Fourier 1822, p. 525), (Fourier \& Freeman 1878, p. 408), although what would be considered a proof by modern standards was not given until much later (Titchmarsh 1948, p. 1). The functions $f$ and $\hat{f}$ often are referred to as a Fourier integral pair or Fourier transform pair.
\\\\
\textbf{Few important properties}:
\begin{itemize}
\item \textbf{Plancherel theorem and Parseval's theorem : } Let $f(x)$ and $g(x)$ be integrable, and let $\hat{f}(\xi)$ and $\hat{g}(\xi)$ be their Fourier transforms. If $f(x)$ and $g(x)$ are also square-integrable, then we have Parseval's theorem (Rudin 1987, p. 187) \cite{Rudin}:

\begin{equation}
\int_{-\infty}^{\infty} f(x) \overline{g(x)} \,{\rm d}x = \int_{-\infty}^\infty \hat{f}(\xi) \overline{\hat{g}(\xi)} \,{\rm d}\xi,
\end{equation}
where the bar denotes complex conjugation.

The Plancherel theorem, which is equivalent to Parseval's theorem, states (Rudin 1987, p. 186)\cite{Rudin}:

\begin{equation}
    \int_{-\infty}^\infty \left| f(x) \right|^2\,{\rm d}x = \int_{-\infty}^\infty \left| \hat{f}(\xi) \right|^2\,{\rm d}\xi. 
\end{equation}

The Plancherel theorem makes it possible to define the Fourier transform for functions in $L^2(\mathbb{R})$. The Plancherel theorem has the interpretation in the sciences that the Fourier transform preserves the energy of the original quantity.
\\\\
\item \textbf{Convolution theorem : } The Fourier transform translates between convolution and multiplication of functions. If $f(x)$ and $g(x)$ are integrable functions with Fourier transforms $\hat{f}(\xi)$ and $\hat{g}(\xi)$ respectively, then the Fourier transform of the convolution is given by the product of the Fourier transforms $\hat{f}(\xi)$ and $\hat{g}(\xi)$ (under other conventions for the definition of the Fourier transform a constant factor may appear).

This means that if:
\begin{equation}
    h(x) = (f \star g)(x) = \int_{-\infty}^\infty f(y)g(x - y)\,dy,
\end{equation}
where $\star$ denotes the convolution operation, then:
\begin{equation}
    \hat{h}(\xi) = \hat{f}(\xi)\cdot \hat{g}(\xi).
\end{equation}

Conversely, if $f(x)$ can be decomposed as the product of two square integrable functions $p(x)$ and $q(x)$, then the Fourier transform of $f(x)$ is given by the convolution of the respective Fourier transforms $\hat{p}(\xi)$ and $\hat{q}(\xi)$. 
\end{itemize}

\subsubsection{2.1.2. General Theory on Circle : }
Let $Z_1, Z_2, \ldots, Z_n$ be the observed data on $S^1$. Let $F$ denote their distribution on $S^1$. Note that each $Z_j = e^{i\theta_j}$ and $F_n = \frac{1}{n}\sum_{i = 1}^{n}\delta_{Z_i}$. Also
\begin{equation}
\hat{F}_n(k) = \int_0^{2\pi}e^{ikx}dF_n(x) = \frac{1}{n}\sum_{i = 1}^{n}e^{ik\theta_j} = \frac{1}{n}\sum_{i=1}^{n}Z_i^k = \hat{u}_k = \frac{1}{n}\sum_{j=1}^{n}\{\cos(k\theta_j) + i\sin(k\theta_j)\}
\end{equation}
Clearly, $\hat{F}_n(0) = 1$ and $\hat{F}_n(-k) = \overline{\hat{F}_n(k)}$.
Before we move ahead, we first prove that there is a 1-1 correspondence between the data and the empirical moments. 

\subsubsection{Proposition 1.} \textit{Empirical Moments have a 1-1 correspondence with the Data.}
\begin{proof}
Suppose we know all the empirical moments, i.e., we know, $\sum Z_i, \sum Z_i^2, \ldots, \sum Z_i^n$. We need to find all the data. Note the data, $Z_1, \ldots, Z_n$ forms the root of the polynomial $\prod_{i=1}^{n}\left(z -Z_i\right)$. We can find all the coefficients of the power of $z$ from the empirical moments. For eg. the coefficient of $z^2$ is $\sum_{i \not= j}Z_iZ_j$, which can be calculated from the moments, as $\sum_{i \not= j}Z_iZ_j = (\sum Z_i)^2 - \sum Z_i^2$. Thus, if we know all the empirical moments we can compute the coefficient of all powers of $z$ of the polynomial. Hence, the roots of the polynomial will be the data. Conversely, if we know the data, we can obviously find the empirical coefficient. Thus, there is a 1-1 correspondence between the two. Hence proved.
\begin{flushright}$\blacksquare$\end{flushright}
\end{proof}

Now, given a positive finite Borel measure $\mu$ on the real line $\mathbb{R}$, the Fourier transform $Q$ of $\mu$ is the continuous function
\begin{equation}
Q(t) = \int_{\mathbb{R}} e^{-itx}d \mu(x).
\end{equation}    
$Q$ is continuous since for a fixed $x$, the function $e^{-itx}$ is continuous and periodic. The function $Q$ is a positive definite function, i.e. the kernel $K(x, y) = Q(y - x)$ is positive definite; 
Now Bochner's Theorem says that the converse if true.	

Thus in case of data on the circle, by the Bochner's Theorem, the moment transformation $E(Z^k)$, transforms the problem from the measure space to the sequence space of bi-infinite sequences which are non-negative definite. Note that the solution to the Smoothing Spline problem in circle, belongs to the class of non-negative definite sequences. Thus Bochner's Theorem gives a 1-1 correspondence between non-negative definite probability measure and non-negative definite sequences. Now all square summable non-negative definite sequences are equivalent to all measures which are absolutely continuous with respect to the Haar measure satisfying $\int |f''|^2 d\lambda < \infty$. This result gives rise to Fourier Splines. Let us explain how.

	First note that the smoothing spline problem, is actually a curve fitting problem, where the minimization takes place on the space of all functions with square integrable second derivative, correspondingly on the space of non-negative definite sequences when the data is coming from a circle. However, our problem is actually density estimation, where we have a data function $Y(x)$ and we need a fitted function $\hat\mu(x)$ such that the $L^2$ norm is minimized, that is, we need to minimize 
\begin{equation}
\frac{1}{2\pi}\int_0^{2\pi}| Y(x) - \hat\mu(x) |^2 dx 
\end{equation}. 

We know that the empirical Fourier coefficients is a sequence in $\mathbb{C}$. Thus, given data on the circle, we can transform $P_n \to \hat{P}_n  = (u_0,u_1,u_2,\ldots ) \in \mathbb{C}$. Similarly, the density can be transformed, $f \to \hat{f} = (\mu_0, \mu_1, \mu_2, \ldots)$. Now, if we are in the situation, where we know that the density belongs to $L^2(0,2\pi)$, then using the empirical Fourier coefficients we can derive the density using the inverse Fourier transformations, i.e. $\hat{f}(x) \to f(x) = \sum u_je^{ijx}$. 
Note that we can write $L^2(0,2\pi)$ as the class of functions, $\mathcal{F}_n$ defined by
\begin{equation}
\mathcal{F}_n = \left\lbrace f : f(x) = \sum_{j = -n/2}^{n/2} u_j e^{-ijx} \textnormal{ where } (u_0,u_1,\ldots) \in \mathbb{C} \textnormal{ and } \sum|u_i|^2 < \infty\right\rbrace,
\end{equation}
where we have consider the coefficients to be 0 for $|j| > \frac{n}{2}$. Thus our aim is to get a function belonging to this class. 

	Now because of the form of the Fourier coefficients on the circle, the minimization problem given in (23) reduces to $\sum_{k \neq 0}|\hat{u_k} - u_k|^2$. However, as $k$ runs from $-\infty$ to $\infty$, this sum diverges to $\infty$. Thus we work with only the $n$ moments, to get a solution belonging to $L^2(0,2\pi)$. That is we try to minimize, $\sum_{k = -n/2}^{n/2}|\hat{u_k} - u_k|^2$. Note that the empirical Fourier coefficients is not absolutely continuous with respect to the Haar measure. Hence by Bochner's Theorem, the solution to this problem derived as the empirical Fourier coefficients is not square summable.  Thus, we must add the penalty term to the problem, viz, $\lambda \|u\|^2$ to get a square summable solution. Now, square summability of the solution implies the existence of a density $f$ by the Fourier Inversion Theorem \cite{Feller}.
	
Thus the Fourier coefficient $u_k$ can now be written as 
\begin{equation}
u_k = x_k + iy_k = \frac{1}{2\pi}\int_0^{2\pi}e^{ikx}f(x)dx
\end{equation}
Note that $u_{-k} = \overline{u_k}$. Furthermore, by inverse Fourier transformation, and considering that the density is real, we can say
\[\begin{split}
f(x)&= \sum_{k=-\infty}^{\infty}u_ke^{-ikx} = 1 + \sum_{k=1}^{\infty}u_ke^{-ikx} + \sum_{k=-\infty}^{-1}u_ke^{-ikx}\\
&= 1 + \sum_{k = 1}^{\infty}\{(x_k + iy_k)e^{-ikx} + (x_k - iy_k)e^{ikx}\}\\
&= 1 + 2\sum_{k=1}^{\infty}\textnormal{Re}\{(x_k + iy_k)e^{-ikx}\}\\
\end{split}
\]
Thus we have, 
\begin{equation}
f(x) = 1 + 2\sum_{k=1}^{\infty}\{x_k\cos(kx) + y_k\sin(kx)\}
\end{equation}	

Hence now, the penalty term become $\lambda\int f^{''}(x)^2dx$, and the problem now becomes, to find the minimizer of 
\begin{equation}
\frac{1}{2\pi}\int_0^{2\pi}| Y(x) - f(x) |^2 dx  + \lambda\int f^{''}(x)^2dx
\end{equation}
in the class $L^2(0,2\pi)$. Equivalently, the problem in the smoothing spline format becomes, the minimization problem of
\begin{equation}
\sum_{k = -\infty}^{\infty}|\hat{u_k} - u_k|^2 + \lambda\{\sum_{k = 1}^{\infty}k^4(x_k^2 + y_k^2)\}
\end{equation}

Thus, we get a 1-1 correspondence between the curve fitting problem and the density estimation problem using the Fourier basis. Thus the Fourier Spline technique is a method of solving the Smoothing Spline problem, where the solution is obtained using the Fourier basis. Now we proceed to actually solving this minimization problem.

\subsection{Solution to the Minimization Problem}
\begin{theorem}
Solution to the minimization problem $\sum_{k = -\infty}^{\infty}|\hat{u_k} - u_k|^2 + \lambda\{\sum_{k = 1}^{\infty}k^4(x_k^2 + y_k^2)\}$ gives rise to a kernel like density estimate $\hat{f}_n(x) = 1 + \sum_{|k| = 1}^{\infty} \hat{u}_kC_k(\lambda) e^{-ikx}$, where $C_k(\lambda) = \frac{1}{1 + \lambda k^4}$, obtained by the convolution of the empirical Fourier coefficients, with the kernel $K(x) = 1 + \sum_{|k|=1}^{\infty}C_k(\lambda)e^{-ikx} = 1 + 2\sum_{k=1}^{\infty}\frac{\cos kx}{1 + \lambda k^4}$ for $x \in [-\pi,\pi)$
\end{theorem}
\begin{proof}
From equation (26) we get,
\begin{equation}
f_n(x) = 1 + 2\sum_{k=1}^{n/2}\{x_k\cos(kx) + y_k\sin(kx)\}
\end{equation}
Differentiating this twice, we get,
\begin{equation}
f_n^{''}(x) = -2\sum_{k=1}^{n/2}\{x_kk^2\cos(kx) + y_kk^2\sin(kx)\}
\end{equation}
We need to minimize $\sum_{k = -n/2}^{n/2}|\hat{u_k} - u_k|^2 + \lambda\{\sum_{k = 1}^{n/2}k^4(x_k^2 + y_k^2)\}$. That is 
\[\sum_{k=1}^{n/2}\bigg\{\left(\overline{\cos(k\theta)} - x_k\right)^2 + \left(\overline{\sin(k\theta)} - y_k\right)^2\bigg\} + \lambda\bigg\{\sum_{k = 1}^{n/2}k^4(x_k^2 + y_k^2)\bigg\}
\]
Differentiating this with respect to $x_k$ and $y_k$ for each $k$, and equating to zero, we get,
\begin{equation}
\widehat{x_k} = \frac{\overline{\cos(k\theta)}}{1 + \lambda k^4} \qquad \textnormal{and} \qquad \widehat{y_k} = \frac{\overline{\sin(k\theta)}}{1 + \lambda k^4}
\end{equation}
Now let $C_k(\lambda) = \frac{1}{1 + \lambda k^4}$, Thus the estimated Fourier coefficient is $\hat{u}_kC_k(\lambda)$. Now $f(x) = 1 + \sum_{|k| = 1}^{\infty} u_ke^{-ikx}$. Thus $\hat{f}_n(x) = 1 + \sum_{|k| = 1}^{\infty} \hat{u}_kC_k(\lambda) e^{-ikx}$, which is the seen as the convolution with the kernel $K(x) = 1 + \sum_{|k|=1}^{\infty}C_k(\lambda)e^{-ikx} = 1 + 2\sum_{k=1}^{\infty}\frac{\cos kx}{1 + \lambda k^4}$ for $x \in [-\pi,\pi)$.
\\\\ 
Hence, proved. 
\begin{flushright}$\blacksquare$\end{flushright}
\end{proof}

\subsubsection{Remark :} Note that our kernel may not be non-negative, for very small values of $\lambda$. The Figure 4, shows plots of the kernel $K(x)$, defined as,
\begin{equation}
K(x) = 1 + 2\sum_{k=1}^{\infty}\frac{\cos kx}{1 + \lambda k^4} \;\; \textnormal{for}\;\; x \in [-\pi,\pi)
\end{equation}
for different values of $\lambda$. It can be seen from the figure that for $\lambda < 0.8$, the Kernel does take negative values. However, such situations are not out of the ordinary. A study of the literature reveals that such kernels present in the class of higher order kernels have been used for bias reduction in kernel density estimation \cite{Jones,Jones2}. They have also been used in non-parametric curve estimation as they often give faster asymptotic rates of convergence. \cite{Marron}.

\begin{figure}[h]
\centering
\includegraphics[scale=0.8]{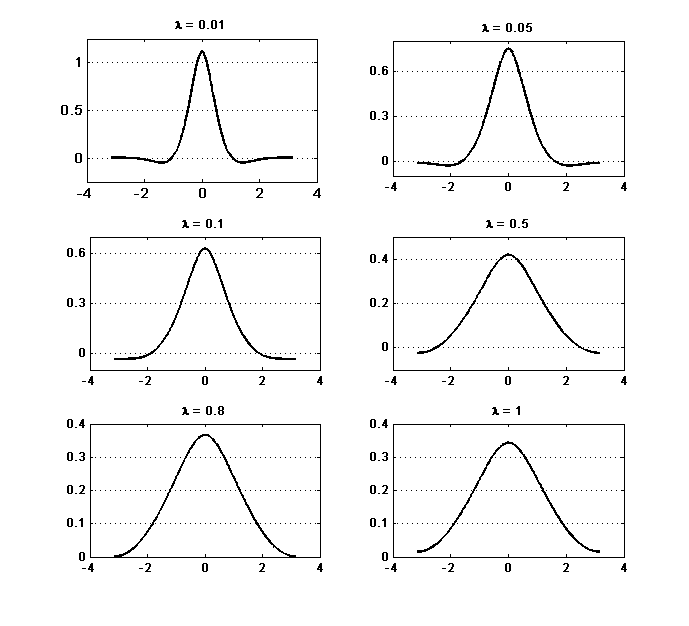}
\caption{The Kernel with different values of $\lambda$. Note that for values of $\lambda$ more than $0.8$ we get a non-negative kernel. But for smaller $\lambda$, the kernel does take non-positive values, as seen from the above Figures.}
\end{figure}

With this estimate of density at hand, we proceed to the calculation of the mean square error as well as the optimal rate and choice of the penalty parameter $\lambda$.

\begin{theorem}
For the density estimate obtained above, we have,
\[
\begin{split}
\textnormal{E}\big|\hat{f}_n(x) - f(x)\big|^2 &=  \sum_k\sum_{k'}u_{k}\bar{u}_{k'}(C_{k}(\lambda) - 1)(C_{k'}(\lambda) - 1)e^{-i(k-k')x}\\
&+ \sum_k\sum_{k'}\frac{C_{k}(\lambda)C_{k'}(\lambda)}{n}\left(u_{k-k'} - u_k\bar{u}_{k'} \right)e^{-i(k-k')x}\\
\end{split}
\]
\end{theorem}

\begin{proof}
From the expression of $\hat{f}_n$ obtained from Theorem 2, we get, 
\[
\begin{split}
\big|\hat{f}_n(x) - f(x)\big|^2 &= \big\{\sum_{|k| = 1}^{\infty} \left( \hat{u}_kC_k(\lambda)- u_k\right) e^{-ikx} \big\}\times \overline{\big\{\sum_{|k| = 1}^{\infty} \left( \hat{u}_kC_k(\lambda)- u_k\right) e^{-ikx} \big\}}\\
&=\sum_k\sum_{k'}\left(\hat{u}_kC_k(\lambda) - u_k\right)\overline{\left(\hat{u}_{k'}C_{k'}(\lambda) - u_{k'}\right)}e^{-i(k-k')x}\\
\textnormal{E}\big|\hat{f}_n(x) - f(x)\big|^2 &= \sum_k\sum_{k'}\textnormal{E}\left(\hat{u}_kC_k(\lambda) - u_k\right)\overline{\left(\hat{u}_{k'}C_{k'}(\lambda) - u_{k'}\right)}e^{-i(k-k')x}\\
&=\sum_k\sum_{k'}\textnormal{E}\left(\hat{u}_kC_k(\lambda) - u_k\right)\left(\overline{\hat{u}_{k'}}C_{k'}(\lambda) - \overline{u_{k'}}\right)e^{-i(k-k')x}\\
&=\sum_k\sum_{k'}\textnormal{E}\left[\left(\frac{1}{n}\sum_{j=1}^{n}Z_j^k\right)C_k(\lambda) - u_k\right]\left[\left(\frac{1}{n}\sum_{j=1}^{n}\bar{Z}_j^{k'}\right)C_{k'}(\lambda) - \overline{u_{k'}}\right]e^{-i(k-k')x}\\
\end{split}
\]
Now the term inside the expectation can be written as
\[
\textnormal{E}\left[\bigg\{\frac{1}{n}\sum_{j=1}^{n}\left(Z_j^k - u_k\right) + u_k\bigg\}C_k(\lambda) - u_k\right]\overline{\left[\bigg\{\frac{1}{n}\sum_{j=1}^{n}\left(Z_j^{k'} - u_{k'}\right) + u_{k'}\bigg\}C_{k'}(\lambda) - u_{k'}\right]}
\]
Replacing $\frac{1}{n}\sum_{j=1}^{n}\left(Z_j^k - u_k\right)$ by $R_k$, the expression reduces to 
\[
\begin{split}
&\textnormal{E}\{\left(R_k + u_k\right)C_k(\lambda) - u_k\}\overline{\{\left(R_{k'} + u_{k'}\right)C_{k'}(\lambda) - u_{k'}\}}\\
&=\textnormal{E}\left(R_kC_k(\lambda) + u_k(C_k(\lambda) - 1)\right)\left(\bar{R}_{k'}C_{k'}(\lambda) + \bar{u}_{k'}(C_{k'}(\lambda) - 1)\right)\\
&=C_{k}(\lambda)C_{k'}(\lambda)\textnormal{E}(R_kR_{k'}) + u_{k}(C_{k}(\lambda) - 1)\bar{u}_{k'}(C_{k'}(\lambda) - 1)\\
&=\frac{1}{n}C_{k}(\lambda)C_{k'}(\lambda)\left(u_{k-k'} - u_k\bar{u}_{k'} \right) + u_{k}(C_{k}(\lambda) - 1)\bar{u}_{k'}(C_{k'}(\lambda) - 1)\\
\end{split}
\]
Thus we get,
\[
\begin{split}
\textnormal{E}\big|\hat{f}_n(x) - f(x)\big|^2 &=  \sum_k\sum_{k'}u_{k}\bar{u}_{k'}(C_{k}(\lambda) - 1)(C_{k'}(\lambda) - 1)e^{-i(k-k')x}\\
&+ \sum_k\sum_{k'}\frac{C_{k}(\lambda)C_{k'}(\lambda)}{n}\left(u_{k-k'} - u_k\bar{u}_{k'} \right)e^{-i(k-k')x}\\
\end{split}
\]
Hence Proved.
\begin{flushright}$\blacksquare$\end{flushright}
\end{proof}

\begin{theorem}
The optimal rate of the MISE is $n^{-4/5}$, which matches the rate for the original kernel smoothing problem.
\end{theorem}
\begin{proof}
Using the expression for the mean square error, we calculate the integrated mean square error as,
\[\begin{split}
\int_0^{2\pi}\textnormal{E}\big|\hat{f}_n(x) - f(x)\big|^2dx &=  \sum_k\sum_{k'}u_{k}\bar{u}_{k'}(C_{k}(\lambda) - 1)(C_{k'}(\lambda) - 1)\int_0^{2\pi}e^{-i(k-k')x}dx\\
&+ \sum_k\sum_{k'}\frac{C_{k}(\lambda)C_{k'}(\lambda)}{n}\left(u_{k-k'} - u_k\bar{u}_{k'} \right)\int_0^{2\pi}e^{-i(k-k')x}dx\\
&=2\pi\sum_ku_{k}\bar{u}_{k}(C_{k}(\lambda) - 1)^2 + 2\pi\sum_k\frac{C_k(\lambda)^2}{n}(1 - u_k\bar{u}_k)\\
&=2\pi\sum_k|u_{k}|^2(C_{k}(\lambda) - 1)^2 + 2\pi\sum_kC_k(\lambda)^2\frac{(1 - |u_{k}|^2)}{n}\\
\end{split}
\]
First note that if $\lambda \to 0, C_k(\lambda) \to 1.$ which implies the first term goes to zero while the second term diverges to infinity. Thus we have with us a situation which is similar to the bias-variance trade off situations. In order to find the optimal rate we to equate the rates of these two terms. Now note that the 2nd term of the above expression can be written as 
\[2\pi\sum_kC_k(\lambda)^2\frac{(1 - |u_{k}|^2)}{n} = \frac{\sum_kC_k(\lambda)^2}{n}A,
\]
where
\begin{equation}
A = 2\pi\sum_kC_k(\lambda)^2\frac{(1 - |u_{k}|^2)}{\sum_jC_j(\lambda)^2}
\end{equation}
which can be thought of an expectation of a random variable which takes values $2\pi(1 - |u_k|)^2$ with probability, $\frac{C_k(\lambda)^2}{\sum_kC_k(\lambda)^2}$. Since this is bounded, the rate depends on $\frac{\sum_kC_k(\lambda)^2}{n}$, which is clearly $1/\left(n\lambda^{1/4}\right)$.
Similarly when we work with the first term, we see that
\[\sum_k(C_k(\lambda) - 1)^2|u_k|^2 = \sum_k\frac{\lambda^2k^4}{(1+\lambda k^4)^2} k^4|u_k|^2 = \sum_k k^4|u_k|^2 \times\sum_k\frac{\lambda^2k^4}{(1+\lambda k^4)^2} p_k
\]
where
\begin{equation}
p_k = \frac{k^4|u_k|^2}{\sum_k k^4|u_k|^2}
\end{equation}
Note that the term $\sum_k k^4|u_k|^2 = B$(say) is an approximation to $\int f^{''}(x)dx$. Thus the 1st term can be written as
\begin{equation}
2\pi\sum_k(C_k(\lambda) - 1)^2|u_k|^2 = \lambda\times B\sum_k\frac{\lambda k^4}{(1+\lambda k^4)^2} p_k
\end{equation}
which has a decay rate of $\lambda$. Thus the optimal rate is obtained by equating these two rates.
\begin{equation}
\frac{1}{n\lambda_n^{1/4}} = \lambda_n \Rightarrow \lambda_n = n^{-4/5}
\end{equation}
Hence Proved.
\begin{flushright}$\blacksquare$\end{flushright}
\end{proof}
Furthermore, the optimal MISE can be expressed as,
\begin{equation}
MISE_{opt}(\hat{f}) = 2\pi\sum_k|\hat{u}_{k}|^2(C_{k}(\hat\lambda) - 1)^2 + 2\pi\sum_kC_k(\hat\lambda)^2\frac{(1 - |\hat{u}_{k}|^2)}{n}
\end{equation}
where $\hat\lambda$ is obtained numerically over by a grid search which minimizes the above equation.
With this methodology, we are able to estimate the density function at a point $x$ using the optimal penalty as $\hat\lambda$, obtained by the above mentioned procedure. Before we show the simulated results on the various different samples, we briefly describe the methodology for using the $\cos^2$ kernel.

\subsection{Using $\cos^2$ Kernel}
Let $Z_1, Z_2, \ldots, Z_n$ be the observed data on $S^1$. Let $K$ denote a smoothing distribution on $S^1$. Let $\xi_1, \xi_2, \ldots, \xi_n$ be i.i.d. $K$. Define $Z_1' = Z_1\xi_1, Z_2' = Z_2\xi_2, \ldots, Z_n' = Z_n\xi_n$. Note that
\[P(Z\xi \in A) = \frac{1}{n}\sum_iP(\xi \in \bar{Z_i}A) = \frac{1}{n}\sum_i\int_{\bar{Z_i}A}K(u)d\mu(u)
\]
Thus the density estimate at some point $a$ is given by $\frac{1}{n}\sum_iK(\bar{Z_i}a)$. Note that we can transform each $Z_i$ to $\theta_i \in [0,2\pi)$ and each point $a$ to $e^{ix}$. Thus for each $x \in [0,2\pi)$ the density estimate is $\frac{1}{n}\sum_jK(e^{i(x - \theta_j})$. Let $\tilde{K}(x) = K(e^{ix})$. One such choice of $\tilde{K}(x)$ can be defined as follows
\[\tilde{K}(x) = \begin{cases}
\frac{2m\cos^2(mx)}{\pi} & \textnormal{if } |mx| \leq \frac{\pi}{2}\\
0 & \textnormal{otherwise}\\
\end{cases}
\]
Since note that $\int_{-\pi/2m}^{\pi/2m}\cos^2(mx) = \frac{\pi}{2m}$. Thus the density estimate at point $x$ boils down to
\[\hat{f_n}(x) = \frac{1}{n}\times\frac{2m}{\pi}\sum_{i=1}^{n}\cos^2\{m(x - \theta_i)\}I_{|m(a-\theta_i)|\leq \pi/2}
\]
Another estimate using inverse Fourier transformation, i.e. for $Z\xi_m$ we get,
\[f(t) = \frac{1}{2\pi}\sum_{k = -\infty}^{\infty}e^{-itk}\left(\frac{1}{n}\sum_{i = 1}^{n}Z_i^k\right) E(\xi_m^k)
\] where the expectation is calculated through the cosine density.

\subsection{Mean Square Error Calculations}
We have performed detailed comparative simulation study to see which method produces the best results. The methodologies followed for the comparative study are as follows :
\begin{itemize}
\item Using Fourier Spline Methodology
\item Using the $\cos^2$ kernel using plug-in technique for bandwidth
\item Using usual KDE with cross-validation technique for bandwidth selection
\item Using usual KDE with plug-in technique for bandwidth selection
\end{itemize}

We have considered a wrapped mixture normal model on the circle with two modes at $\theta$ and $-\theta$, where $\theta \in [0,\pi/2)$. We take $n = 50$ observations and calculate the density using the four methodologies. The Epanechnikov kernel has been used in the usual KDE methodologies. Since in the method based on Fourier Spline, the smoothing parameter $\lambda$ for the kernel, does not work as the usual linear scaling parameter, we consider the $\cos^2$ kernel as an alternative comparison by using the usual plug-in estimate. We repeat this procedure for $N = 10^6$ times and calculate the mean squared errors. All simulations have been performed on MATLAB 7.8.0. The results of the simulation study are shown in Table 2.

It can be seen that when we are considering a bi-modal continuous distribution, the estimate obtained through the Fourier Spline has a smaller MSE when compared to all other methods. 

\begin{table}[h]
\centering
\caption{Mean Square Errors at $x = 0$}
\begin{tabular}{|@{\hspace{0.25cm}}c@{\hspace{0.25cm}}|@{\hspace{0.25cm}}c@{\hspace{0.25cm}}|@{\hspace{0.25cm}}c@{\hspace{0.25cm}}|@{\hspace{0.25cm}}c@{\hspace{0.25cm}}|@{\hspace{0.25cm}}c@{\hspace{0.25cm}}|}\hline
$\theta$ & Fourier & $\cos^2$ & KDE Cross- & KDE \\
Values & Spline & Kernel & Validation & Plug-in\\
\hline
0.10&0.0041&0.0056&0.0086&0.0101\\
0.15&0.0033&0.0049&0.0069&0.0077\\
0.20&0.0041&0.0053&0.0116&0.0103\\
0.25&0.0028&0.0038&0.0086&0.0105\\
0.30&0.0029&0.0037&0.0087&0.0107\\
0.35&0.0040&0.0053&0.0099&0.0085\\
0.40&0.0033&0.0042&0.0119&0.0095\\
0.45&0.0022&0.0028&0.0061&0.0083\\
0.50&0.0023&0.0029&0.0070&0.0042\\
0.55&0.0028&0.0034&0.0080&0.0094\\
0.60&0.0033&0.0037&0.0095&0.0100\\
0.65&0.0025&0.0029&0.0063&0.0055\\
0.70&0.0018&0.0023&0.0044&0.0067\\
0.75&0.0021&0.0021&0.0057&0.0060\\
0.80&0.0022&0.0025&0.0047&0.0071\\
0.85&0.0019&0.0020&0.0046&0.0074\\
0.90&0.0002&0.0019&0.0047&0.0059\\
0.95&0.0027&0.0027&0.0057&0.0087\\
1.00&0.0024&0.0025&0.0058&0.0101\\
1.05&0.0018&0.0019&0.0028&0.0047\\
1.10&0.0016&0.0017&0.0040&0.0041\\
1.15&0.0018&0.0018&0.0041&0.0101\\
1.20&0.0014&0.0013&0.0023&0.0043\\
1.25&0.0017&0.0018&0.0026&0.0025\\
1.30&0.0014&0.0015&0.0022&0.0029\\
1.35&0.0015&0.0016&0.0023&0.0023\\
1.40&0.0015&0.0016&0.0019&0.0026\\
1.45&0.0013&0.0012&0.0015&0.0028\\
1.50&0.0016&0.0017&0.0018&0.0039\\
1.55&0.0015&0.0015&0.0017&0.0100\\
\hline
\end{tabular}

\end{table}

\section{Detection of Local Features}
Let $\mathcal{T} = [0, 2\pi) = \bigcup A_{nj}$ where $A_{nj} = \left[\frac{2\pi j}{2^n},\frac{2\pi(j+1)}{2^n} \right)$ for $j = 0,1,\ldots, 2^n - 1$ is a partition of the circle. Note that $2^nA_{nj} - 2\pi j = [0, 2\pi)$. Now for $y \in [0,1)$, $ f\left(\frac{2\pi(j+y)}{2^n}\right)$ is the local information of $f$ on $A_{nj}$. $\mathbb{P}_n$ is a measure on $[0, 2\pi)$. Let the sample from $A_{nj}$ be denoted by $S_{nj}$ and $S_n = \bigcup_{j=0}^{2^n-1} S_{nj}$. Thus $2^nS_{nj} - 2\pi j$ maps $S_{nj}$ to $[0, 2\pi)$. As $n$ becomes larger the information becomes more and more local. 
	
As mentioned before, we can see that because of Theorem 2, the localisation property is lost as the density is calculated by a weighted average of these local information. Hence the usual Fourier Spline estimate does not give a good estimate of the local features. Thus, in this section we develop the methodology of detecting the local features. Estimation of the local features is explained in details in Section 4.
	
	We use the information present in each $A_{nj}$ to detect the presence of discontinuity or edges. We shall now delve into the details of those procedure starting with the generalised testing procedure. The distributional details for testing specific local features are explained subsequently.

\subsection{The Testing Procedure}
We first fix an $n$ large enough. Suppose we want to test
\[ H_0 : \textnormal{There is some local feature in } A_{nj} \textnormal{ vs. } H_1 : H_0 \textnormal{ is false}
\]

The exact testing procedure will depend on the local feature that we are trying to determine and is explained later in sub-sections 3.2 - 3.3. Suppose $T$ be the test statistic, and based on this we can find the $p$-value of the test, let us denote that by $p_{A_{nj}}$. Note that the same test can be performed on the previous layer, that is, on the partition formed by $A_{(n-1)j}$ and subsequently on earlier layers. Thus we have a nested layers culminating in the last layer. It might be possible to get contradictions, for e.g. the test is rejected at partition level $n$, it might be possible that the test is accepted at partition level $n-1$. So, we can consider having a multivariate $p$-value profile, where the $i$-th co-ordinate is the $p$-value corresponding to testing on $A_{ij}$. This poses a problem of inferring from a multivariate $p$-value.

We suggest a procedure for the same. Note that if we had a single $p$-value for all the units in the partition we could perform the Holm's procedure of multiple testing \cite{Holms} to give us the regions which contain local features at a some level $\alpha$. However instead of a single $p$-value, for each unit in the partition we have a multivariate $p$-value profile, say $(p_1,...,p_k)'$. Now under null, the $p$-value usually follows a uniform distribution, but under the alternative, it can be modelled well using the $Beta(\alpha,\beta)$ distribution with a high value for $\alpha$. Thus we can assume, $p_1, \ldots, p_k$ are i.i.d $Beta(\alpha,\beta)$. Thus an obvious sufficient statistic is the geometric mean, i.e. $(\prod_{i=1}^kp_i)^{1/n}$. However instead of using the usual geometric mean in the Holms procedure, we use a weighted geometric mean, where the weights are chosen to be in a decreasing sequence with the increase in the partition level. We came to this conclusion after extensive simulation study. That is, for each $j$, we have with us a $p$-value profile $(p_1, \ldots, p_k)$, and we consider the sufficient statistic as 
\begin{equation}
P_j = \prod_{i=1}^{k}p_i^{w_i}
\end{equation}
where $\sum_{i=1}^kw_i = 1$ and $\left\lbrace w_i\right\rbrace_{1 \leq i \leq k} $ is a decreasing sequence, for $j = 0,\ldots, 2^n - 1$. Using, these $P_0, \ldots, P_{2^n - 1}$ we apply the Holm's procedure which can be explained as follows.
\subsubsection{3.1.1 Holm's Procedure}
The Holm procedure \cite{TSH} can conveniently be stated in terms of the $p$-values $P_0, \ldots, P_{2^n - 1}$ of the $2^n = s$(say) individual tests. Let the ordered $p$-values be denoted by $P_{(1)} \leq \ldots \leq P_{(s)}$, and the associated hypotheses by $H_{(1)}, \ldots , H_{(s)}$. Then the Holm procedure is defined stepwise as follows:
\\\\
\textit{Step 1.} If $P_{(1)} \geq \alpha/s$, accept $H_{(1)}, \ldots , H_{(s)}$ and stop. If $P_{(1)} < \alpha/s$ reject $H_{(1)}$ and test the remaining $s - 1$ hypotheses at level $\alpha/(s - 1)$.
\\\\
\textit{Step 2.} If $P_{(1)} < \alpha/s$ but $P_{(2)} \geq \alpha/(s - 1)$, accept $H_{(2)}, \ldots , H_{(s)}$ and stop. If $P_{(1)} < \alpha/s$ and $P_{(2)} < \alpha/(s - 1)$, reject $H_{(2)}$ in addition to $H_{(1)}$ and test the remaining $s - 2$ hypotheses at level $\alpha/(s - 2)$. And so on.
\\\\

Thus by this procedure we shall be able to detect the exact regions which contain the local features. Now we describe the exact tests for detecting the individual local features starting with support and outlier detection.

\subsection{Support and Outlier Detection}
This is the first step towards the detection of local features. Note that we have with us the data which is most likely following a mixture distribution from a discrete and a continuous distribution. So it is likely to get a few discrete observations on the edges of the support. Our aim is to find all such positions so that we can eliminate such outliers and then later incorporate them again to give the final estimate. 

Now when we consider each $A_{nj}$, observe that if the support begins at some point $x_0 \in A_{nj}$ (without loss of generality we consider the support to be the right of $x_0$) then we expect a very small number of points to the left of $x_0$, while an abundance of points to the right of $x_0$. Thus we could model the left of $x_0$ by $B(n,c/n)$ which converges to a Poisson number of points, while the right side follows a multinomial distribution, with probability directly proportional of the length of the arc. Note that here the optimal test would be testing $H_0 : \lambda = 0$ vs $H_1: \lambda > 0$ (where $\lambda$ is the parameter of the Poisson Model). However, testing this situation is very risky, because $\lambda = 0$ denotes a degenerate distribution at 0. Hence, if somehow we get a single value to the left of $x_0$, our test is rejected. In order to go around this situation, we consider the $B(n,c/n)$ and test for $H_0 : c \leq 1$ vs. $H_1 : c > 1$. 

Hence our algorithm for testing
$ H_0 : \textnormal{End of upport lies in } A_{nj} \textnormal{ vs. } H_1 : H_0 \textnormal{ is false},
$ is as follows : 
\begin{description}
\item[Step 1 :] Denote $x_0$ as the centre of $A_{nj}$
\item[Step 2 :] Count the no. of obsevations to the left of $x_0$, call it $T$.
\item[Step 3 :] Accept the Null hypothesis if $T \leq 1$, otherwise reject it.
\item[Step 4 :] The $p$-value of the test is $1 - F(T)$, where $F$ is the cumulative probability distribution for $B(n,1/n)$
\end{description}
Note that this exact same test works for testing $ H_0 : \textnormal{There exists outliers in } A_{nj} \textnormal{ vs. } H_1 : H_0 \textnormal{ is false}.$ Note that we do this test for every partition layer and using the procedure mentioned in section 3.1 we finally get in the which units of the partition, the test is rejected at level $\alpha$. Thus, with  this information, we are able to identify the exact region of the support and the points of outliers. So, we keep track of such regions and remove them from the data. Thus now we work in the support of the distribution and proceed to the detection of discontinuities and edges.

\subsection{Discontinuity and Edge Detection}
Note that the regions of discontinuity or edges can be suitably modelled using the exponential family,
\begin{equation}
f(x|\boldsymbol{\beta}) \propto \exp\left\lbrace\beta_0 x + \beta_1I(x>x_0) + \beta_2(x - x_0)_+\right\rbrace
\end{equation}
where $x_0$ is the mid point of the region $A_{nj}$ in question, $I$ is the indicator function and $(\cdot)_+$ is a positive function which take the argument value if it is positive and 0 otherwise. Note that here $\beta_1$ is the parameter which controls the jump discontinuity, while $\beta_2$ controls the edge effects.
We first test the existence of discontinuity in a particular unit of the partition, say $A_{nj}$. If, there is no discontinuity we proceed to testing if the $A_{nj}$ has an edge or not. 
Thus the algorithm for testing the presence of discontinuity or edge is as follows:
\begin{description}
\item[Step 1 :] We first test $H_0 : $ There is no discontinuity in $A_{nj}$ vs. $H_1 : H_0$ is false. This is equivalent to testing $H_0 : \beta_1 = 0$ vs $H_1 : \beta_1 \not= 0$.
\item[Step 2 :] We proceed by the usual Likelihood ratio test for exponential family. Our test statistic is 
\begin{equation}
\Lambda = \frac{\sup_{\beta \in H_0} \prod_{x_i \in S_{nj}} f(x_i|
\boldsymbol{\beta})}{\sup_{\beta \in H_0 \bigcup H_1} \prod_{x_i \in S_{nj}}f(x_i|\boldsymbol{\beta})}
\end{equation} 
Note that the supremum is easily obtained by using the MLE of the parameters which are easy to obtain experimentally since both the numerator and the denominator belong to the exponential family. \cite{TPE}
\item[Step 3 : ] Now $T = -2\log\Lambda$ follows a Chi-Squares distribution with degrees of freedom $3 - 2 = 1$. Thus, we reject the test if $T > \mathcal{X}_{1;\alpha/2}$, where $\mathcal{X}_{1;\alpha/2}$ denotes the $(100 - \alpha/2)$-th percentile of the chi-squares distribution with degrees of freedom 1.
\item[Step 4 : ] If the test is rejected, that is, there is discontinuity in this region, we stop, and report the $p$-value as $1 - F(T)$, where $F$ is the cumulative probability distribution of the chi-squares with degrees of freedom 1. If no discontinuity is detected we proceed to Step 5.
\item[Step 5 : ] Now we test $H_0 : $ There is no edge in $A_{nj}$ vs. $H_1 : H_0$ is false. This is equivalent to testing $H_0 : \beta_2 = 0$ vs $H_1 : \beta_2 \not= 0$. 
\item[Step 6 : ] We proceed exactly in same way as in Steps 2 and 3. 
\item[Step 7 : ] Finally we report the $p$-value as $1 - F(T)$, where $F$ is the cumulative probability distribution of the chi-squares with degrees of freedom 1.
\end{description}

As before, we do this for every partition layer and using the procedure in Section 3.1, we finally detection the regions which contain discontinuity or edges. Thus we are able to break our support into compact intervals, which can be clearly categorised into sets containing local features such as discontinuity or edges and sets where the density is smooth. Using this knowledge we can proceed to the overall density estimate as will be explained in Section 4. 

\section{Overall Density Estimates}
In this section we describe how we obtain the overall density estimate using the idea of partition of unity. Note that using the methodology described in Section 3, we will be able to detect regions of local features, in the support of the density as well as regions of outliers. From these regions we can create an open cover of the support of the density. Thus based on the theory of Partition of Unity, we can find a partition using the function,
\begin{equation}
\rho(x) = \begin{cases}
\exp\left(-\tan^2\frac{\pi x}{2\sigma}\right) & -\sigma < x < \sigma\\
0 & \textnormal{otherwise}
\end{cases}
\end{equation}
Note that this function is smooth and takes positive values in $(-\sigma,\sigma)$ and 0 outside this support. Thus with appropriate choice of $\sigma$ we can always construct a partition of unity, say $\{\rho_i\}_{i \in I}$ using this function. We shall use this partition of unity to estimate the density in the concerned regions.

Note that as explained before, suppose we find $n$ regions which possess a local feature, and another region where the density is smooth. Let $\{U_i\}$ denote these open sets, and $\{\rho_i\}_{i \in I}$ denote the partition, such that supp $\rho_i \subseteq U_i$. Using this notation, our whole density can be partitioned as
\begin{equation}
f(x) \propto \sum_{i = 1}^{n}\rho_i(x)g(x) + \left(1 - \sum_{i=1}^{n}\rho_i(x)\right)h(x)
\end{equation}
where each $\rho_i(x)g(x)$ estimates the local feature in the set $U_i$, while $\left(1 - \sum_{i=1}^{n}\rho_i(x)\right)h(x)$ corresponds to the smooth region which is estimated using the Fourier Spline technique as explained in Section 2. Note that $\rho_i$ takes the value 0 ourside $U_i$, thus we can find the normalising constant by integration. Thus our final density estimate is
\begin{equation}
f(x) = \frac{\sum_{i = 1}^{n}\rho_i(x)g(x) + \left(1 - \sum_{i=1}^{n}\rho_i(x)\right)h(x)}{\int_0^{2\pi}\sum_{i = 1}^{n}\rho_i(x)g(x) + \left(1 - \sum_{i=1}^{n}\rho_i(x)\right)h(x)\; dx}
\end{equation}
Thus we start by estimating the local feature.
\subsection{Estimation of Local Feature}
We have with us a region which is the support of $\rho_i$ and we wish to estimate the local feature $\rho_i(x)g(x)$. Consider all the points lying in the support of $\rho_i$, i.e. consider all $x \in U_i$. We fit the local exponential model given by,
\begin{equation}
g(x|\boldsymbol{\beta}) = \frac{\exp\left\lbrace\beta_0 x + \beta_1I(x>x_0) + \beta_2(x - x_0)_+\right\rbrace}{\int_{U_i}\exp\left\lbrace\beta_0 x + \beta_1I(x>x_0) + \beta_2(x - x_0)_+\right\rbrace}
\end{equation}
using the data in $U_i$. Let $\int_{U_i}\exp\left\lbrace\beta_0 x + \beta_1I(x>x_0) + \beta_2(x - x_0)_+\right\rbrace = A(\beta_0,\beta_1,\beta_2)$, (say). Then the log-likelihood function can be written as 
\begin{equation}
l(\beta_0,\beta_1,\beta_2) = \beta_0\sum_{x \in U_i}x + \beta_1\sum_{x \in U_i}I(x > x_0) + \beta_2\sum_{x \in U_i}(x - x_0)_+ - n\log A(\beta_0,\beta_1,\beta_2)
\end{equation}
where $n$ is the cardinality of $U_i$. From this, the maximum likelihood estimates of $(\beta_0,\beta_1,\beta_2)$ can be obtained, numerically. Thus, the final estimate for the local feature at $U_i$ is 
\begin{equation}
\rho_i(x)\times\frac{\exp\left\lbrace\hat\beta_0 x + \hat\beta_1I(x>x_0) + \hat\beta_2(x - x_0)_+\right\rbrace}{A(\hat\beta_0,\hat\beta_1,\hat\beta_2)}
\end{equation}
where $(\hat\beta_0,\hat\beta_1,\hat\beta_2)$ denotes the maximum likelihood of the parameters. This procedure is repeated for all the regions which have been detected to possess local features. Now we proceed to adapting the method based on Fourier Splines to estimate the smooth portion of the density.

\subsection{Estimation of Smooth Feature}
We are left with estimating the smooth part of the density, i.e. $\left(1 - \sum_{i=1}^{n}\rho_i(x)\right)f(x)$. Suppose we have the data as $Z_1, Z_2, \ldots, Z_m$ lying on the circle. Treating the required function as a smooth function which we need to estimate, we can see that the empirical Fourier coefficient is given by,
\begin{equation}
F_m(k) = \int_{0}^{2\pi}e^{ikx}\left(1 - \sum_{i=1}^{n}\rho_i(x)\right)f(x)\;dx = \int_{0}^{2\pi}e^{ikx}\left(1 - \sum_{i=1}^{n}\rho_i(x)\right)\;dF_m(x)
\end{equation}
Thus, this simplifies to,
\begin{equation}
\hat{u}_k  = \frac{1}{m} \sum_{j=1}^{m}Z_j^k\left(1 - \sum_{i=1}^{n}\rho_i(Z_j)\right)
\end{equation}
Thus, instead of the usual average of the $Z_i^k$ we get a weighted average of $Z_i^k$. The sample moment, in this situation becomes weighted sample moments, where the weights are proportional to where the data is coming from. If $Z_j$ comes from a region which has a local feature, say $U_i$, then its weight in the smooth region decreases from $1/m$ to $(1 - \rho_i(Z_j))/m$. On the other hand, if $Z_j$ belongs to the smooth region, then by the definition of the partition $\{\rho_i\}$ we have, $\rho_i(Z_j) = 0\; \forall\; i \in [1,n]$. Thus the weight remains as $1/m$. 

Now following the same procedure as in Section 2 we get the estimate of the smooth function as
\begin{equation}
\hat{f}_m(x) = 1 + \sum_{|k| = 1}^{m/2} \hat{u}_kC_k(\lambda) e^{-ikx}
\end{equation}
where 
\[
\hat{u}_k = \frac{1}{m} \sum_{j=1}^{m}Z_j^k\left(1 - \sum_{i=1}^{n}\rho_i(Z_j)\right)\;\; \textnormal{and}\;\; C_k(\lambda) = \frac{1}{1 + \lambda k^4}
\]
The optimal choice of $\lambda$ is obtained by minimizing the MISE given by equation (37). Thus the final estimate is 
\begin{equation}
\hat{f}(x) = \frac{\sum_{i=1}^{n}\rho_i(x)\hat{g}_i(x) + \hat{f}_m(x)}{\int_{0}^{2\pi}\sum_{i=1}^{n}\rho_i(x)\hat{g}_i(x) + \hat{f}_m(x)}
\end{equation}
where $\rho_i(x)\hat{g}_i(x)$ and $\hat{f}_m(x)$ is obtained from (46) and (49). This completes the theory of estimation of density on the circle, taking into account the local features. We now proceed to the section on simulations to highlight how our methodology works in different situations.

\section{Simulation Study}
In this Section we apply our methodology, on three test data sets, each showing a particular aspect of our Methodology. The Usual Kernel Density estimates are also applied on the data sets, and a comparative study of their MISE is provided.
\subsection{Density with Discontinuity and Edge Effects}
We use the same density which we used to portray the existence of such local features in Section 1. Our data is coming from a mixture distribution of a Uniform and a Triangular Density, with equal mixing proportions. Thus, we consider $X_1, X_2, \ldots, X_n$ coming from a equal mixture of $Unif\left(\frac{-3\pi}{4},\frac{-\pi}{4}\right)$ and $Triangular\left(\frac{\pi}{4},	\frac{3\pi}{4}\right)$ with a peak at $\frac{\pi}{2}$. We take $n = 1000$ and execute the algorithm, starting with the detection of support and outliers. 

	We have broken the circle into $2^9$ regions and we work with the data on these intervals. First the detection of support, gives us almost the exact regions, starting from a little left of $\frac{-3\pi}{4}$ to a little right of $\frac{-\pi}{4}$, and  from a little left of $\frac{\pi}{4}$ to a little right of $\frac{3\pi}{4}$. We then proceed to the detection of discontinuities and edges. We get two points of discontinuity and two edge points. Breaking the data into the rest of the region, we estimate the smooth region using the Fourier Spline technique, and fit the exponential model to estimate the regions showing the local features. The final estimated density is shown in Figure 5 along with the usual kernel density estimate. The whole procedure is repeated $10^5$ to find an estimate of the mean integrated squared error. We also calculate the MISE of the estimates of the local features, which is tabulated in Table 3.

\begin{figure}[h]
\centering
\includegraphics[scale=0.55]{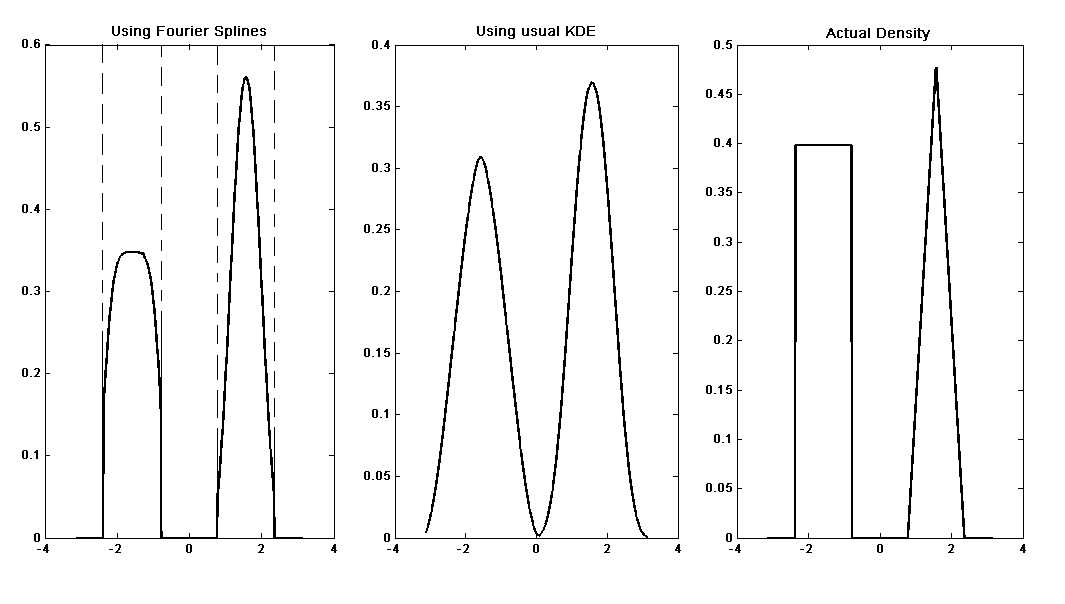}
\caption{The Fourier Spline Estimate along with the Kernel Density estimate and the Actual Density.}
\end{figure}
	
\begin{table}[h]
\centering
\caption{Mean Integrated Square Errors}
\begin{tabular}{|@{\hspace{0.25cm}}c@{\hspace{0.25cm}}|@{\hspace{0.25cm}}c@{\hspace{0.25cm}}|@{\hspace{0.25cm}}c@{\hspace{0.25cm}}|@{\hspace{0.25cm}}c@{\hspace{0.25cm}}|@{\hspace{0.25cm}}c@{\hspace{0.25cm}}|}\hline
Estimates & Using Fourier & Using usual & Local Feature & Local Feature \\
&Spline  & KDE & at $-3\pi/4$ & at $-\pi/4$\\
\hline
MISE & 0.0058 & 0.0112 & 0.0049 & 0.0051\\
\hline
\end{tabular}
\end{table}

Now, note that, the support of the distribution has been very well identified in estimate through Fourier Splines. Furthermore, all points of discontinuities and edge effects has been correctly identified, allowing us to apply the smooth function estimation technique on the remaining set. The vertical dotted lines show the points detected as having local features. Finally we estimate the local features using the exponential model. The MISE of these local estimates are also small as seen from Table 3. Note that such positions of local features have not been identified in the estimate obtained from the standard KDE procedure. Thus, the estimate from the Fourier Spline is indeed much more closer to the actual density. This claim is further justified from Table 3, where we see that the MISE from the Fourier Spline is much smaller than the MISE from the usual KDE.	

\subsection{Mixture Density of a Continuous and Discrete Distribution}	
Here we consider data from a mixture density of a normal distribution with outliers coming from a discrete distribution. We have considered $X_1, X_2, \ldots X_n$ to come from a mixture of a normal distribution with a mode at $0$ truncated to $[-\frac{\pi}{2},\frac{\pi}{2})$ with probability $1 - \epsilon$, and with probability $\epsilon$ from a discrete distribution taking values $\{-3\pi/4 , 3\pi/4\}$ with equal probability. We have chosen $\epsilon = 0.05, n = 1000$ for our simulations. Exactly similar procedure is carried out to reach the final estimate shown in Figure 6, along with the usual kernel density estimates. The whole procedure is repeated $N = 10^5$ times to get an estimate of the mean integrated squared error. We also calculate the MISE of the estimates of the local features. Both results are tabulated in Table 4.

\begin{figure}[h]
\centering
\includegraphics[scale=0.55]{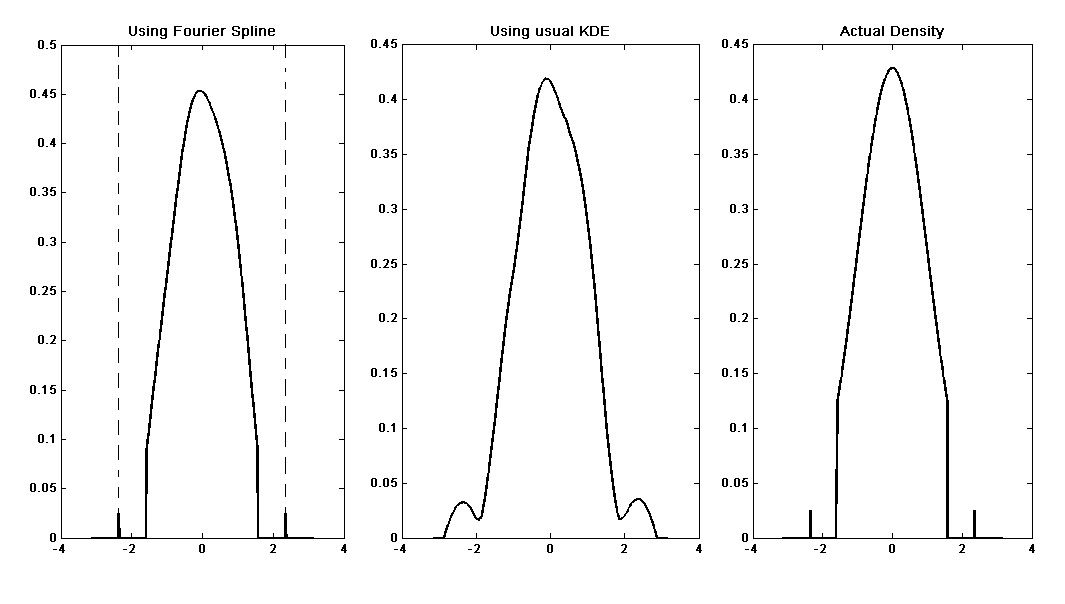}
\caption{The Fourier Spline Estimate along with the Kernel Density estimate and the Actual Density.}
\end{figure}

\begin{table}[h]
\centering
\caption{Mean Integrated Square Errors}
\begin{tabular}{|@{\hspace{0.25cm}}c@{\hspace{0.25cm}}|@{\hspace{0.25cm}}c@{\hspace{0.25cm}}|@{\hspace{0.25cm}}c@{\hspace{0.25cm}}|@{\hspace{0.25cm}}c@{\hspace{0.25cm}}|@{\hspace{0.25cm}}c@{\hspace{0.25cm}}|}\hline
Estimates & Using Fourier & Using usual & Local Feature & Local Feature \\
&Spline  & KDE & at $-3\pi/4$ & at $3\pi/4$\\
\hline
MISE & 0.0041 & 0.0145 & 0.0009 & 0.0010\\
\hline
\end{tabular}
\end{table}

As before, the support has been very well detected, which allows us to accurately estimate the smooth function, which is not the case from the KDE as seen from the second figure. Furthermore, the outliers has been almost perfectly detected as can be seen through the dotted lines in Figure 6. Those regions have been estimated through the local exponential family. Note that the presence of the outliers creates humps in the estimate from the KDE, which fails to identify the discrete nature. The very fact, which motivated us to work on this problem, and we have been able to successfully create a work about for the problem. As before we can see that our estimate through Fourier Spline, has a smaller MISE when compared to that from the standard KDE.

\subsection{Density with Several Discontinuities}
We consider another example where there are more than 2 discontinuities in order to check whether our methodology is successful in detecting them all or not. So we consider a data coming from a density which can be written as
\begin{equation}
f(x) = \begin{cases}
\frac{1}{\pi}	& \textnormal{ for } x \in \left[-\frac{\pi}{2},-\frac{\pi}{4}\right)\\
\frac{1}{2\pi}	& \textnormal{ for } x \in \left[-\frac{\pi}{4},\frac{\pi}{4}\right)\\
\frac{2}{\pi}	& \textnormal{ for } x \in \left[\frac{\pi}{4},\frac{\pi}{2}\right)
\end{cases}
\end{equation}
This density has four points of discontinuity, viz, $-\frac{\pi}{2}, -\frac{\pi}{4}, \frac{\pi}{4}$ and $\frac{\pi}{2}$. We draw $n = 1000$ observations from this density and perform the procedure as done in the last two examples. The final density estimate is given in Figure 7. The whole procedure is repeated for $N = 10^5$ times to 
get the mean integrated square error. MISE for the estimated local features are also calculated. Results of which are tabulated in Table 5. 

\begin{figure}[h]
\centering
\includegraphics[scale=0.415]{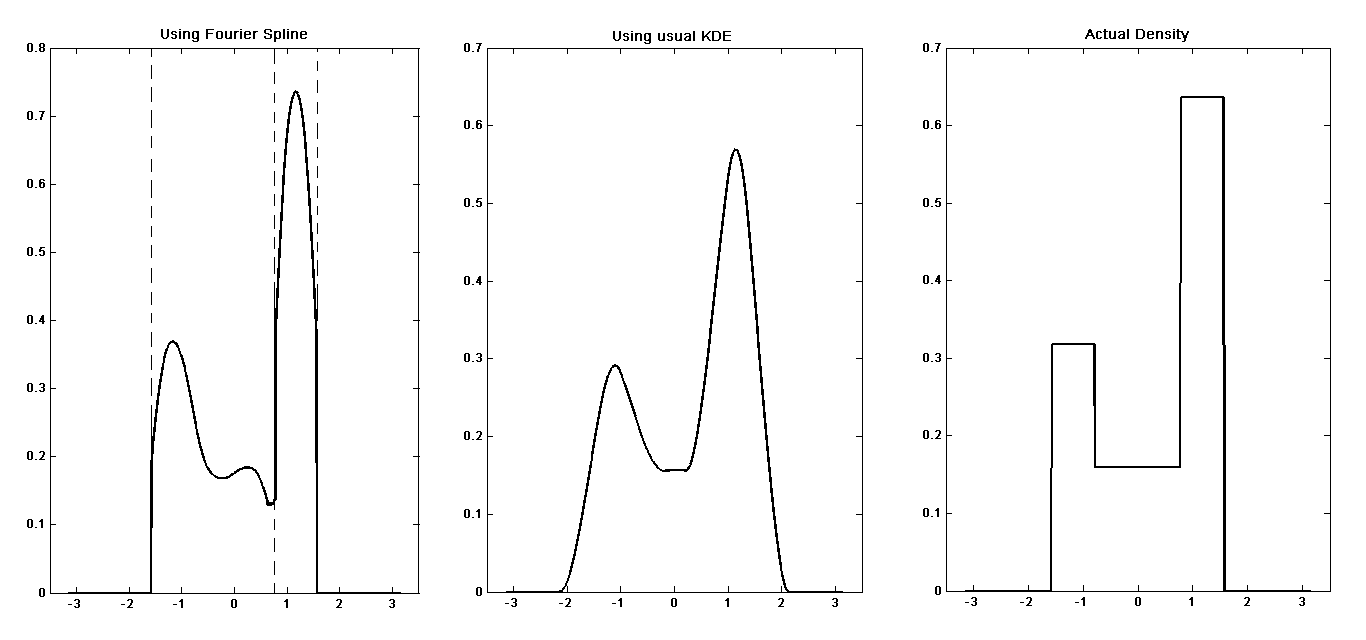}
\caption{The Fourier Spline Estimate along with the Kernel Density estimate and the Actual Density.}
\end{figure}

\begin{table}[h]
\centering
\caption{Mean Integrated Square Errors}
\begin{tabular}{|@{\hspace{0.25cm}}c@{\hspace{0.25cm}}|@{\hspace{0.25cm}}c@{\hspace{0.25cm}}|@{\hspace{0.25cm}}c@{\hspace{0.25cm}}|@{\hspace{0.25cm}}c@{\hspace{0.25cm}}|@{\hspace{0.25cm}}c@{\hspace{0.25cm}}|
@{\hspace{0.25cm}}c@{\hspace{0.25cm}}|}\hline
Estimates & Using Fourier & Using usual & Local Feature & Local Feature & Local Feature \\
&Spline  & KDE & at $-\pi/2$ & at $\pi/4$ & at $\pi/2$\\
\hline
MISE & 0.0036 & 0.0083 & 0.0020 & 0.0031 & 0.0027\\
\hline
\end{tabular}
\end{table}

Here too, the support of the distribution has been perfectly detected. Note further that the 3 dotted lines in the first figure of Figure 7, shows that our algorithm has successfully detected 3 local features, mainly discontinuities at $-\pi/2, \pi/4$ and $\pi/2$. However, it has failed to detect the discontinuity at $\pi/4$. This can be explained by the fact that the jump discontinuity in modulus is small at $-\pi/4$ when compared to the other jump discontinuities in the density. Thereby, its chances of detection is indeed very small. However, our algorithm has been able to detect discontinuities with a moderately high jump. 

Now on when we look at the estimate through the Kernel density estimation, we see that it has neither been able to detect the support nor the local features accurates. As a results we can see from Table 5 that the MISE of our method, here too, (although it has failed to detect a discontinuity) is far better than MISE of the estimate obtained through the standard method of Kernel Density Estimation.

 \section{Real Life Data}
Here we apply our methodology on a real life data obtained from Sengupta and Rao (1966) \cite{Sengupta}. They have provided the data coming from Cross-bedding Azimuths in the 3 Units of Kamthi River. This data is given in Table 6. We have used the data from Upper Kamthi, to show our Methodology. Since we are given grouped data, we have considered in each group they follow a uniform distribution, with minor rotation in each bin to get more uniformity. The histogram of the data, along with the Fourier Spline and standard KDE estimates are shown in Figure 8.

\begin{table}[h]
\centering
\caption{Frequency Distributions of Cross-bedding Azimuths in the 3 Units of Kamthi River (Sengupta and Rao, 1966)}
\begin{tabular}{|@{\hspace{0.25cm}}c@{\hspace{0.25cm}}|@{\hspace{0.25cm}}c@{\hspace{0.25cm}}|@{\hspace{0.25cm}}c@{\hspace{0.25cm}}|@{\hspace{0.25cm}}c@{\hspace{0.25cm}}|}\hline
Azimuth (in degrees) & Lower Kamthi & Middle Kamthi & Upper Kamthi\\
\hline
0 - 19&14&50&75\\
20 - 39&14&62&75\\
40 - 59&11&33&15\\
60 - 79&13&9&25\\
80 - 99&9&1&7\\
100 - 119&16&3&3\\
120 - 139&0&0&3\\
140 - 159&4&0&0\\
160 - 179&0&0&0\\
180 - 199&3&0&0\\
200 - 219&4&2&21\\
220 - 239&0&8&8\\
240 - 259&0&0&24\\
260 - 279&0&11&16\\
280 - 299&6&5&36\\
300 - 319&7&20&75\\
320 - 339&1&53&90\\
340 - 359&21&41&107\\
\hline
\end{tabular}
\end{table}

\begin{figure}[h]
\centering
\includegraphics[scale=0.434]{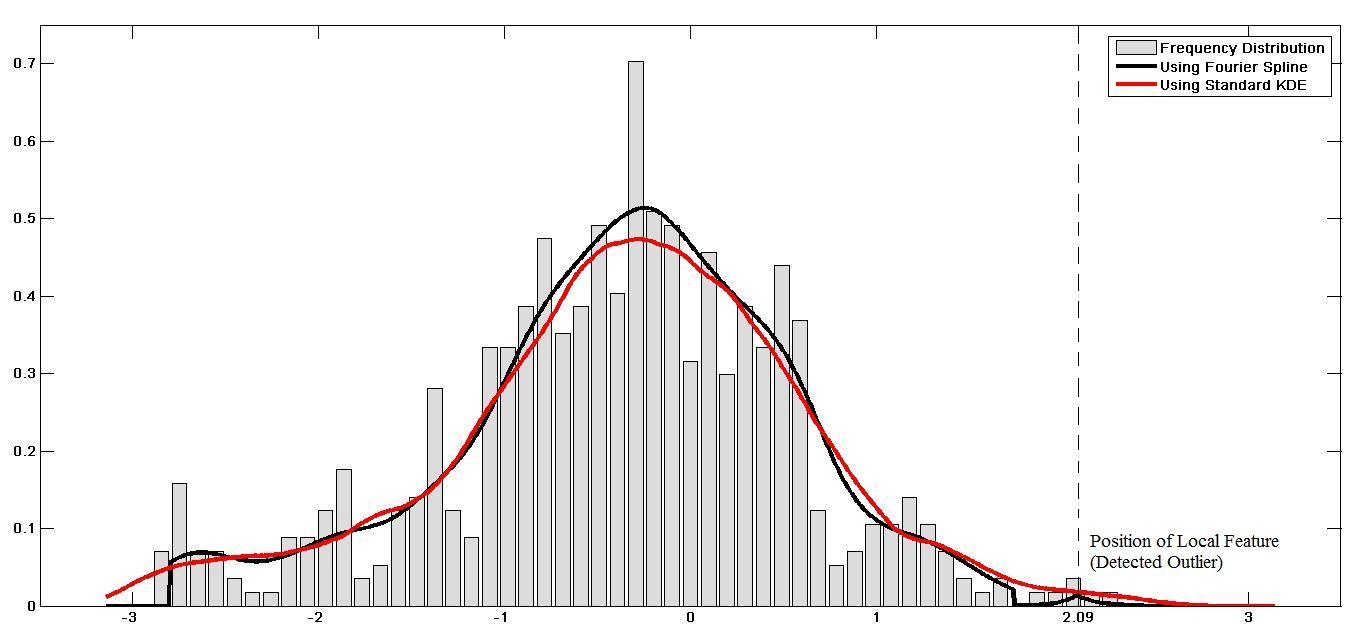}
\caption{Application of Fourier Spline Methodology on Real Data}
\end{figure}

In Figure 8, the gray line depicts the estimate from the standard KDE, while the black depicts the estimate from Fourier Spline. Here too we have been able to almost perfectly detect support. The extreme regions which have no data, has been correctly identified. Furthermore, few local features such as outliers has also been identified by the procedure at 2.09. As seen from the figure, our estimate is much more closer to the actual probability distribution than the standard KDE.

\section{Discussions}

\subsection{Regarding the Spline Smoothing on Circles.}
Note that although we have got an estimate of our density function by solving the Spline Smoothing on Circle, our solution is not the  exact solution because, the solution of the Spline Smoothing problem must belong to the space of non-negative definite sequence. If our solution was indeed non-negative definite, then by Bochner's Theorem, the corresponding Kernel would have been non-negative. But as seen from Figure 4  in Section 2, this is not the case. Thus, solving the Spline problem on circles with the non-negative definite constraint is one of the open problems which should influence researchers to further work in this field.

\subsection{Regarding Estimation of the Local Features}
In this article, we have estimated the local features using the exponential model, and then multiplied it with the corresponding function $\rho(x)$ where $\rho \in \{\rho_i\}_{i \in I}$ , a partition of unity, created on the set of open covers using the function
\begin{equation}
\rho(x) = \begin{cases}
\exp\left(-\tan^2\frac{\pi x}{2\sigma}\right) & -\sigma < x < \sigma\\
0 & \textnormal{otherwise}
\end{cases}
\end{equation}
for different choices of $\sigma$. Let us explain this in a bit more details. Suppose, we have detected a local feature in $A_{nj}$ (using the same notation as before), then $A_{nj}$ forms a unit of the open cover. The rest can be chosen so that they form an intersecting set. Now, suppose, $\rho_i$ is has its support in $A_{nj}$. So, in order to estimate the local feature, what we have done is we have first estimated the exponential model using the data in $A_{nj}$ and then multiplied that with $\rho_i$. The final estimates is done through proper normalisation. 

	Now instead of doing this, another way of estimating the local feature could be estimating the whole function $\rho_i(x) \times g(x)$, where $g(x)$ a density from the chosen exponential family. That is, we can find the optimal values of $(\beta_0,\beta_1,\beta_2)$ by maximizing the likelihood given by
	
\begin{equation}
L(\beta_0,\beta_1,\beta_2) = \prod_{x_j \in S_{nj}} \rho_i (x_j) \times \frac{\exp\left\lbrace\beta_0 x_j + \beta_1I(x_j>x_0) + \beta_2(x_j - x_0)_+\right\rbrace}{\int_{A_{nj}}\exp\left\lbrace\beta_0 x + \beta_1I(x>x_0) + \beta_2(x - x_0)_+\right\rbrace}
\end{equation}	
instead of maximising,
\begin{equation}
 L(\beta_0,\beta_1,\beta_2) = \prod_{x_j \in S_{nj}} \frac{\exp\left\lbrace\beta_0 x_j + \beta_1I(x_j>x_0) + \beta_2(x_j - x_0)_+\right\rbrace}{\int_{A_{nj}}\exp\left\lbrace\beta_0 x + \beta_1I(x>x_0) + \beta_2(x - x_0)_+\right\rbrace}
\end{equation}
This would give us another estimate and we could do a comparative study on these different methods to see which method gives the best estimate of the local features. 

\section{Conclusion}
	
In this article, we have developed a new methodology for estimating Circular probability distributions, using Spectral Isomorphism to solve an equivalent Spline Smoothing Problem on the Circle.  
	
	There has been already a lot of literature which aims at density estimation on the Circle however most of them are based on Kernel Density Estimation with different modifications. But, working on the problem with the knowledge of presence of local features has hardly been attempted. Since during smoothing of the kernel density the local feature gets lost, it is very difficult to capture these using the conventional kernel density estimation.
	
	We develop the methodology based on the usual Spline Smoothing problem, which tries to solve the Curve Fitting problem. In our case, because of the form of the Fourier coefficient on the circle, the minimisation problem for finding the density in the $L^2$ space, gets transformed to the Spline Smoothing problem on the Circle. The solution to that problem, via Bochner's Theorem can be transformed to a Kernel which has a square integrable second derivative. Thus, the density can be obtained through the inverse Fourier transformation. We have been able to solve the Spline Smoothing problem on the circle and also derive several results on the resultant density estimate. Estimates for the Bias and MSE has been calculated. The optimal rate of the density has been proved to match the optimal rate in the standard Kernel Density Estimation. Finally our methodology has been compared with other standard methods as a comparative study.
	
	Another major focus of this work has been in the detection and estimation of local features such that outliers, discontinuity and edges. The testing problem is different for the concerned local feature however, the overall multiple testing has been performed using a novel application of the Holm's procedure as well as a novel technique of working with a multivariate $p$-value profile.
	
	The Support and Outlier detection has been done by considering a $B(n,c/n)$ model, while discontinuity and edge detection has been performed by considering an exponential family. The final overall estimate is obtained by considering a partition of unity of a special function. The local functions has been estimated through exponential model and the smooth function is through a weighted adaptation of the Fourier Spline methodology. 
	
Lastly, we have applied our technique to a few artificial densities and to a real life data as well. In all situations, this technique proved to be a much better estimate than the standard kernel density estimates. We also have a few open problems for further research, especially, the problem of solving the Spline problem on circle, restricted to the non-negative definite solution. As well as the comparative study of the different techniques of estimating the local features. We are also focusing our attention to finding the estimate of the local functions in the Fourier paradigm itself. A few details regarding this and the complete flowchart has been discussed in the Appendix.

\section*{Appendix}
\subsection*{Detecting Local Features using the Fourier Paradigm}

Suppose $f$ denotes the density of the distribution. Let $\mathbb{P}$ denote the measure on the space. Correspondingly $\mathbb{P}_n$ denotes the empirical cdf. Note that by the Riesz Representation Theorem, $\mathbb{P}$ can be thought of as an operator working on $f$ such that 
\begin{equation}
\mathbb{P}(f) = \int f\textnormal{d}\mathbb{P}
\end{equation}
Let $\mathcal{T} = [0, 2\pi) = \bigcup A_{nj}$ where $A_{nj} = \left[\frac{2\pi j}{2^n},\frac{2\pi(j+1)}{2^n} \right)$ for $j = 0,1,\ldots, 2^n - 1$ is a partition of the circle. Note that $2^nA_{nj} - 2\pi j = [0, 2\pi)$. Now for $y \in [0,1)$, $ f\left(\frac{2\pi(j+y)}{2^n}\right)$ is the local information of $f$ on $A_{nj}$. $\mathbb{P}_n$ is a measure on $[0, 2\pi)$. Let the sample from $A_{nj}$ be denoted by $S_{nj}$ and $S_n = \bigcup_{j=0}^{2^n-1} S_{nj}$. Thus $2^nS_{nj} - 2\pi j$ maps $S_{nj}$ to $[0, 2\pi)$. As $n$ becomes larger the information becomes more and more local. 
	
Define $\mathbb{P}_{nj} = \sum_{x \in S_{nj}}\delta(x - .)$, as the empirical cdf on $A_{nj}$. Thus we get,
\begin{eqnarray}
\mathbb{P}_{nj}(\textbf{1}) &=& \int_\mathcal{T}\textnormal{d}\mathbb{P}_{nj} = |S_{nj}|\\
\mathbb{P}_{nj}(f) &=& \sum_{x \in S_{nj}}f(x)
\end{eqnarray}

Denote the conditional empirical cdf on $A_{nj}$ as $\mathbb{Q}_{nj} = \frac{\mathbb{P}_{nj}}{|S_{nj}|}$. Now consider the $2^n$-th Fourier coefficient, i.e,
\[\begin{split}
\mathbb{P}_n(z^{2^n}) &= \mathbb{P}_n(e^{i2^nx}) = \int_0^{2\pi}e^{i2^nx}\textnormal{d}\mathbb{P}_n\\
&= \sum_{j=0}^{2^n - 1}\int_{A_{nj}}e^{i2^nx}\textnormal{d}\mathbb{P}_n = \sum_{j=0}^{2^n - 1}\int_{\frac{2\pi j}{2^n}}^{\frac{2\pi(j+1)}{2^n}}e^{i2^nx}\textnormal{d}\mathbb{P}_n\\
&= \sum_{j = 0}^{2^n-1}\sum_{\frac{2\pi j}{2^n} \leq x_l < \frac{2\pi(j+1)}{2^n}}e^{i2^nx_l} = \sum_{j = 0}^{2^n-1} |S_{nj}| \mathbb{Q}_{nj}(e^{ix})\\
\end{split}
\]
Thus,
\begin{equation}
\int_0^{2\pi}e^{i2^nx}\textnormal{d}\mathbb{P}_n = \sum_{j = 0}^{2^n-1} |S_{nj}| \mathbb{Q}_{nj}(e^{ix})
\end{equation}
where $\mathbb{Q}_{nj}(e^{ix})$ is the scattered empirical Fourier coefficients which capture the localisation property at $A_{nj}$. Thus for each element $A_{nj}$ of the partition we have a localisation value as $\mathbb{Q}_{nj}(e^{ix})$ if $|S_{nj}| > 0$ and $0$ otherwise. We can consider this as the new data, as if we have only these $2^n$ data points, and we try to find the density value at these points. As mentioned before, we can see that because of equation (58) the localisation property is lost as the $2^n$-th Fourier coefficient is calculated by a weighted average of these local information.

Thus, if we can use the information present in each $A_{nj}$ to detect the presence of discontinuity or edges, through $\mathbb{Q}_{nj}(e^{ix})$, and estimate the same, then the whole procedure can be generalized to the Fourier paradigm. We are still trying to accomplish this and hope to succeed in the near future. 

\subsection*{Algorithm for Overall Density Estimation using Fourier Spline Technique}
\begin{description}
\item[Step 1 : ] Considering the given data on the circle, find an $n$ large enough and break the circle into layers of partitions of size $2^{k}$ for $k = 2, \ldots, n$. Call the $j$-th unit of the $k$-th layer as $A_{nj}$.
\item[Step 2 : ] Execute the algorithm for Support and Outlier detection as explained in Section 3.2. Finally perform the Holm's procedure to get the region of the support and the location of the outliers. 
\item[Step 3 : ] Remove the portions which are outside the support, or are outliers. In the remaining portion that perform the Discontinuity and Edge Detection using the algorithm as given in Section 3.3. Finally perform the Holm's procedure to get the locations of the discontinuities. 
\item[Step 4 : ] Keep a note of the $A_{nj}$'s having a discontinuity. Incorporate them during the creation of the open cover of the subset. 
\item[Step 5 : ] Create the partition of unity, by considering the open cover created in the previous step.
\item[Step 5 : ] Estimate the local features such as discontinuities and outliers by estimating the local exponential family, using the data from their corresponding support. The method is described in Section 4.1.
\item[Step 6 : ] Estimate the smooth function by the Fourier Spline Procedure, by considering the weighted empirical coefficients. Details in Section 2 and Section 4.2.
\item[Step 7 : ] Finally the overall estimate is obtained by proper normalisation. The final estimate is give by equation (50).

\end{description}


\begin{thebibliography}{99}
\bibitem{Abra}
Abramovich, F., Benjamini, Y., Donoho, D. L. and Johnstone, I. M. (2006). Adapting to unknown sparsity by controlling the false discovery rate. \textit{Annals of Statistics}. \textbf{34},pp. 584-653.
\bibitem{Abra2}
Abramovich, F., Sapatinas, T. and Silverman, B. W. (1998). Wavelet thresholding via a Bayesian approach. \textit{J. Roy. Statist. Soc. Ser. B} \textbf{60}, 725-749.
\bibitem{Aron}
Aronszajn, N. (1950). Theory of Reproducing Kernels. \textit{Transactions of the American Mathematical Society} \textbf{68 (3)}, pp. 337 - 404.
\bibitem{Berlinet}
Berlinet, A. and Thomas, C. (2004) \textit{Reproducing kernel Hilbert spaces in Probability and Statistics} Kluwer Academic Publishers.
\bibitem{Bouman}
Bouman, C. and Sauer, K. (1992) A Generalized Gaussian Image Model for Edge-Preserving MAP Estimation, \textit{Technical Report TR-EE-92-1}, School of Electrical Engineering, Purdue University.
\bibitem{Chiu}
Chiu, S. (1991) Bandwidth Selection for Kernel Density Estimation
\textit{Ann. Statist.} \textbf{19 (4)},1, pp. 1883-1905. 
\bibitem{Chesneau}
Chesneau, C and Lecue, G. (2009). Adapting to Unknown Smoothness by Aggregation of Thresholded Wavelet Estimators. \textit{Statistica Sinica} \textit{19},pp. 1407-1417.
\bibitem{Donoho2}
Donoho, D.L., and Johnstone, I.M., (1995). Adapting to unknown smoothness via wavelet shrinkage. \textit{Journal of the American Statistical Association}, \textbf{90}(432) pp. 1200 - 1224.
\bibitem{Donoho}
Donoho, D., Johnstone, I., Kerkyacharian, G. and Picard, D. (1996). Density estimation by wavelet thresholding. \textit{The Annals of Statistics}, \textbf{24(2)}, pp. 508 - 539.
\bibitem{Feller}
Feller, W. (1965) \textit{An Introduction to Probability Theory and Its Applications. Volume II.} John Wiley \& Sons, New York, NY. 
\bibitem{Fisher}
Fisher, N.I. (1989). Smoothing a sample of circular data. \textit{Journal of Structural Geology}, \textbf{11}, pp. 775 - 778.
\bibitem{Hall}
Hall, P., Watson, G.S. and Cabrera, J. (1987). Kernel density estimation with spherical data. \textit{Biometrika}, \textbf{74}, pp. 751 - 762.
\bibitem{hastie}
Hastie, T. J.; Tibshirani, R. J. (1990). \textit{Generalized Additive Models}. Chapman and Hall. 
\bibitem{Holms}
Holm, S. (1979). A simple sequentially rejective multiple test procedure. \textit{Scandinavian Journal of Statistics} \textbf{6(2)} pp. 65 - 70.
\bibitem{Jansen}
Jansen, M. (2001). Noise Reduction by Wavelet Thresholding, \textit{Volume 161 of Lectures Notes
in Statistics}. Springer Verlag, New York.

\bibitem{Jones}
Jones, M.C. and Signorini, D. F. (1997) A Comparison of Higher-Order Bias Kernel Density Estimators. \textit{Journal of the American Statistical Association} \textbf{92 (439)} pp. 1063 - 1073.
\bibitem{Jones2}
Jones, M.C. and Foster, P. J. (1993) Generalized jackknifing and higher order kernels. \textit{Journal of Nonparametric Statistics} \textbf{3 (1)}. pp. 81-94.
\bibitem{Jud}
Juditsky, A. (1997). Wavelet estimators: adapting to unknown smoothness. \textit{Math. Methods Statist.} \textbf{1}, 1-20.
\bibitem{Loader}
Loader, C. R. (1999) Bandwidth Selection: Classical or Plug-in? \textit{Ann. Statist.} \textbf{27 (2)}, pp. 415 - 438
\bibitem{Marron}
Marron, J.S. (1994) Visual Understanding of Higher-Order Kernels. \textit{Journal of Computational and Graphical Statistics} \textbf{3 (4)}, pp. 447-458.
\bibitem{Munkres}
Munkres, J. R. (2000) \textit{Topology}. Prentice Hall, NJ, USA. 
\bibitem{Nason}
Nason, G. P. (1995). Choice of the Threshold Parameter in Wavelet Function Estimation, \textit{Wavelets and Statistics, Lecture Notes in Statistics}. Volume 103.
\bibitem{Pelletier}
Pelletier, B. (2005). Kernel density estimation on Riemannian manifolds, \textit{Statistics \& Probability Letters}, \textbf{73 (3)}, 1, pp. 297-304.
\bibitem{Sengupta}
Jammalamadaka, S. Rao and SenGupta, A. (2001). \textit{Topics in Circular Statistics}. World Scientific, Singapore.
\bibitem{TPE}
Lehmann, E. L. and Casella, G. (1998). \textit{Theory of Point Estimation}, Springer. 
\bibitem{TSH}
Lehmann, E. L. and Romano, J. P. (2005). \textit{Testing Statistical Hypotheses}, Springer.
\bibitem{Rudin}
Rudin, W. (1987), \textit{Real and Complex Analysis ($3^{rd}$ ed.)}, Singapore: McGraw Hill.
\bibitem{Silverman}
Silverman, B. W. (1986). \textit{Density Estimation for Statistics and Data Analysis}, Chapman and Hall, London.
\bibitem{Taylor}
Taylor, C.C. (2008). Automatic bandwidth selection for circular density estimation. \textit{Comput. Stat. Data Anal.} \textbf{52 (7)}, pp. 3493-3500. 
\bibitem{Grace}
Wahba, G. (1990) Spline Models for Observational Data. \textit{Society for Industrial and Applied Mathematics}, Philadelphia, PA, USA.


\end{thebibliography}
\end{document}